\documentclass[11pt,a4paper,twoside,openright]{article}

\usepackage[english]{babel}
\usepackage[dvips]{graphicx}
\usepackage[pdftex,colorlinks]{hyperref}
\usepackage{natbib,hypernat}
\usepackage{amsmath,amssymb,amsthm,verbatim,amscd,mathrsfs,amsfonts,latexsym,epic} 
\usepackage{psfrag,epsfig}
\usepackage[usenames]{color}
\usepackage[utf8]{inputenc}

\newtheorem{thm}{Theorem}

\newtheorem{prop}{Proposition}
\newtheorem{cor}{Corollary}

\newtheorem{lem}{Lemma}

\newtheorem{remark}{Remark}
\newtheorem{assum}{Assumption}
\newenvironment{merci}{\textbf{Acknowledgments}}{ }

\newcommand{\pr}{\mathbb{P}}

\newcommand{\esp}{\mathbb{E}}

\newcommand{\Var}{\mathbb{V}\text{ar}}
\newcommand{\Argmin}{\mathop{\mbox{Argmin}}}

\title{On efficient estimators of the proportion
of true null hypotheses in a multiple testing setup}
\author{Van Hanh Nguyen$^{1,2}$ and Catherine Matias$^2$}

\begin{document}
\thispagestyle{empty}

\maketitle

\begin{center}
1.  Laboratoire  de Mathématiques  d'Orsay, Université Paris  Sud, UMR
CNRS 8628, Bâtiment 425, 91~405 Orsay Cedex, France. 
E-mail: nvanhanh@genopole.cnrs.fr \\
2. Laboratoire Statistique et Génome, Université d'Évry Val d'Essonne, 
UMR CNRS 8071- USC INRA, 23 bvd de France, 91~037 Évry, France.
E-mail:  catherine.matias@genopole.cnrs.fr
\end{center}

\begin{abstract}
We consider the problem of estimating the proportion $\theta$ of true null hypotheses in a multiple testing context. The setup is classically modeled through a semiparametric mixture with two components: a uniform distribution on interval $[0,1]$ with prior probability $\theta$ and a nonparametric density $f$. We discuss asymptotic efficiency results and establish that two different cases occur whether $f$ vanishes on a set with non null Lebesgue measure or not. In the first case, we exhibit estimators converging at parametric rate, compute the optimal asymptotic variance and conjecture that no estimator is asymptotically efficient (\emph{i.e.} attains the optimal asymptotic variance). In the second case, we prove that the quadratic risk of any estimator does not converge at parametric rate. We illustrate those results on simulated data. 
\end{abstract}




\section{Introduction}            
The problem of estimating the proportion $\theta$ of true null
hypotheses is of interest in situation where several thousands of (independent)
hypotheses can be tested simultaneously. One of the typical applications in
which multiple testing problems occur is estimating the proportion of
genes that are not differentially expressed in deoxyribonucleic acid
(DNA) microarray experiments \cite[see for instance][]{Dudoit_book}. Among other application domains, we mention astrophysics \citep{Meinshausen_Rice} or neuroimaging \citep{Turkheimer2001}. A reliable estimate of
$\theta$ is important when one wants to control multiple error rates,
such as the false discovery rate (FDR) introduced by
\cite{Benjamini1995}. In this work, we  discuss asymptotic efficiency of estimators of the true proportion of null hypotheses. We stress that the asymptotic framework is particularly relevant in the above mentioned contexts where the number of tested hypotheses is huge. \\

In many recent articles  \citep[such        as][etc]{Broberg2005,
   Celisse-Robin2010,Efron2004,Efron2001,Genovese2004},               a
 two-component  mixture  density is  used  to  model  the behavior  of
 $p$-values $X_1, X_2,\ldots, X_n$ associated with $n$ independent tested hypotheses. More precisely, assume the test statistics are independent
and identically distributed (iid) with a continuous distribution under
the corresponding null hypotheses, then the $p$-values $X_1, X_2,\ldots, X_n$ are iid and follow the
uniform distribution $\mathcal{U}([0,1])$ on interval $[0,1]$ under
the null hypotheses. The density $g$ of
$p$-values is modeled by a two-component mixture with following
expression
\begin{equation}
\label{eq-model}
\forall x \in [0,1], \quad g(x)=\theta + (1-\theta)f(x),
\end{equation}
where $\theta \in [0,1]$ is the unknown proportion of true null hypotheses and
$f$ denotes the density of $p$-values generated under the alternative (false null hypotheses). 

Many different identifiability conditions on the parameter $(\theta,f)$ in model~\eqref{eq-model} have been discussed in the literature. 
For example,  \cite{Genovese2004} introduce the concept of purity  that corresponds to the case where the essential infimum of $f$ on $[0,1]$ is zero. They prove that
purity implies identifiability but not \emph{vice versa}. \cite{Langaas2005}
suppose that $f$ is  decreasing with $f(1)=0$ while \cite{Neuvial2010}
assumes   that  $f$   is  regular   near  $x=1$   with   $f(1)=0$  and
\cite{Celisse-Robin2010}  consider  that   $f$  vanishes  on  a  whole
interval included  in $[0,1]$. These are sufficient  but not necessary
conditions on $f$ that ensure identifiability.  Now, if we assume more
generally that $f$  belongs to some set $\mathcal{F}$  of densities on
$[0,1]$,  then a  necessary  and sufficient  condition for  parameters
identifiability is stated in the  next result, whose proof is given in
Section~\ref{sec:proof_lower}.  

\begin{prop}\label{prop:identif}
The parameter $(\theta,f)$ is identifiable on a set $(0,1)\times
\mathcal{F}$ if and only if for all $f \in \mathcal{F}$ and for all
$c\in (0,1)$, we have $c+(1-c)f \notin \mathcal{F}$.
\end{prop}

This very general result is the  starting point to  considering explicit sets  $\mathcal{F}$ of
densities that ensure  the parameter's identifiability on $(0,1)\times
\mathcal{F}$. In  particular, if $\mathcal{F}$  is a set  of densities
constrained to have essential infimum  equal to zero, one recovers the
purity result of \cite{Genovese2004}.  
However, from an estimation perspective, the purity assumption is very
weak  and it is  hopeless to  obtain a  reliable estimate  of $\theta$
based on the value of $f$ at a unique value (or at a finite number of values). 
In  the following, we  explore asymptotic  efficiency results  for the
estimation of $\theta$ and establish that two different cases are to be distinguished:  models assuming that $f$ vanishes on a set of points
with positive Lebesgue measure and models where this set of points has 
zero measure (and where regularity or monotonicity assumptions are
added  on  $f$).  In  the  first  case,  we  obtain  the  existence  of
$\sqrt{n}$-consistent estimators of  $\theta$ that is to say estimators $\hat \theta_n$ such that $\sqrt{n}(\hat \theta_n-\theta)$ is bounded in probability (denoted by $\sqrt{n}(\hat \theta_n-\theta)=O_{\pr}(1)$). We  exhibit such estimators and also compute the asymptotic optimal variance for this problem. Moreover, we conjecture that asymptotically efficient estimators (that is estimators asymptotically attaining this variance lower bound) do not exist. 
In the second case, while the existence of an estimator $\hat \theta_n$ of $\theta$ converging at parametric rate has   not  been   established   yet,   we  prove   that   if  such   a $\sqrt{n}$-consistent estimator of  $\theta$ exists, then the variance
$\Var(\sqrt{n}\hat \theta_n)$ cannot have a finite limit.  In other words,
the quadratic risk of $\hat \theta_n$ cannot converge to zero at a parametric rate.

Let us now discuss the different estimators of $\theta$ proposed in the
literature, starting with those  assuming (implicitly or not) that $f$
attains its  minimum value on a  whole interval. 
First, \cite{Schweder1982} suggested a procedure to
estimate $\theta$, that has been later used by
\cite{Storey2002}. This estimator depends on an unspecified parameter $\lambda \in
[0,1)$ and is equal to the proportion of $p$-values larger
than  this threshold  $\lambda$  divided by  $1-\lambda$.  It is  thus
consistent  only if  $f$ attains its minimum value   on the  interval  $[\lambda,1]$ (an
assumption not made in the article by \cite{Schweder1982} nor the one
by \cite{Storey2002}). Note that even if such an assumption were made, it would
not solve the problem of choosing $\lambda$ such that $f$ attains its infimum  on $[\lambda,1]$.
Adapting this procedure in order to end up with an estimate of the  positive
FDR (pFDR), \cite{Storey2002} proposes a bootstrap
strategy to pick $\lambda$.  More precisely, his
procedure minimizes the mean squared error for estimating the pFDR. Note that \cite{Genovese2004} established that, for fixed
value $\lambda$ such that the cumulative distribution function (cdf) $F$ of $f$
satisfies $F(\lambda)<1$, \citeauthor{Storey2002}'s
estimator converges  at parametric rate and  is asymptotically normal,
but is also asymptotically biased: thus it does not converge to $\theta$ at parametric rate. Some other choices of $\lambda$ are,
for instance, based on break point estimation \citep{Turkheimer2001}
or  spline  smoothing  \citep{Storey2003}.  Another natural  class  of
procedures  in this  context is  obtained  by relying  on a  histogram
estimator  of $g$  \citep{Mosigetal,Nettleton06}. Among  this  kind of
procedures,    we    mention   the    one    proposed   recently    by
\cite{Celisse-Robin2010}  who  proved  convergence in  probability  of
their  estimator (to the  true parameter  value) under  the assumption
that    $f$    vanishes   on    an    interval.    Note   that    both
\citeauthor{Storey2002}'s  and histogram based estimators  of $\theta$
are constructed using nonparametric  estimates $\hat g$ of the density
$g$ and then  estimate $\theta$ relying on the value of  $\hat g$ on a
specific  interval.  The  main  issue  with  those  procedures  is  to
automatically  select  an  interval  where  the true  density  $g$  is
identically  equal to  $  \theta$.  As a  conclusion  on the  existing
results for this setup ($f$ vanishing on a set of points with non null
Lebesgue measure),  we stress the  fact that none of  these estimators
were  proven to  be  convergent  to $\theta$  at  parametric rate.  In
Proposition~\ref{prop:sqrtn_consistent}  below, we  prove that  a very
simple histogram based estimator  possesses this property, while in Proposition~\ref{prop:CR_sqrtn_consistent}, we establish that this is also true for the more elaborate procedure proposed by~\cite{Celisse-Robin2010} which has the advantage of automatically selecting the "best" partition among a fixed collection. However, we
are not aware of a procedure for estimating $\theta$ that asymptotically attains the optimal variance in this context. Besides, one might conjecture that such a procedure does not exist for regular models (see Section~\ref{sec:one_step}).

Other estimators of $ \theta$  are based on regularity or
monotonicity assumptions made on  $f$ or equivalently on $g$, combined
with  the assumption  that the  infimum of  $g$ is  attained  at $x=1$.  These estimators rely on nonparametric estimates  of  $g$  and  appear  to  inherit  nonparametric  rates  of convergence. \cite{Langaas2005} derive estimators based on nonparametric maximum likelihood
estimation of the $p$-value density, in two setups: decreasing and 
convex decreasing densities $f$.  We mention that no theoretical properties of these
estimators are given. 
\cite{Hengartner1995} propose a very general 
finite sample  confidence envelope for a monotone  density. Relying on
this result and assuming moreover 
that the cdf $G$ of $g$ is concave and 
that $g$ is Lipschitz in a neighborhood of $x=1$, \cite{Genovese2004}
construct an estimator converging to $g(1)= \theta$ at rate $(\log
n)^{1/3}n^{-1/3}$. Under some 
regularity   assumptions  on   $f$   near  $x=1$,   \cite{Neuvial2010}
establishes that by letting $\lambda \to 1$, \citeauthor{Storey2002}'s
estimator may be turned into a consistent estimator of $ \theta$, with a
nonparametric  rate of  convergence  equal to  $n^{-k/(2k+1)}\eta_n$,
where $\eta_n \to +\infty$ and $k$ controls the regularity of $f$ near
$x=1$.  Our   results  are  in   accordance  to  the   literature:  no
$\sqrt{n}$-consistent  estimator  has  been  constructed  yet,  as  is
expected from  the fact  that the quadratic  risk of any  estimator of
$\theta$ cannot converge at parametric rate in this case (see Corollary~\ref{cor:LAM}).

To  finish  this  tour  on  the literature  about  the  estimation  of
$\theta$,   we  mention  that   \cite{Meinshausen_Buhlmann05}  discuss
probabilistic lower bounds for the proportion of true null hypotheses, which are valid under general and unknown dependence structures between the test statistics. 
Finally, note that we do not discuss here estimators of the proportion
of  non null effects in  Gaussian mixtures  such as  in \cite{Cai_Jin10,
  Jin08, Jin_Cai07},  a related but although different  problem as the
one we study. \\

The  article  is organized  as  follows. 
Section~\ref{sec:efficiency} establishes lower bounds on the quadratic
risk      for     the      estimation      of     $\theta$,      while
Section~\ref{sec:sqrtn_consistency}   explores   corresponding   upper
bounds, \emph{i.e.} the  existence of $\sqrt{n}$-consistent estimators
of $\theta$ and the existence of asymptotically efficient estimators.  Section~\ref{sec:simus} illustrates our  results relying
on simulations. The proofs of the main results are postponed to Section~\ref{sec:proofs}, while some technical lemmas are proved in Appendix~\ref{appendix}.

\section{Lower bounds for the quadratic risk and efficiency}\label{sec:efficiency}
In this section, we give lower bounds for the quadratic risk of any estimator of $\theta$. For any fixed unknown parameter $\delta \in [0, 1)$, we introduce a  
set of densities $\mathcal{F}_\delta$ (with respect to the Lebesgue measure $\mu$) and an induced  set of semiparametric distributions $\mathcal{P_\delta}$, respectively defined as  
\begin{align*}
\mathcal{F}_{\delta}  & =\{  f  :  [0,1]\mapsto
\mathbb{R}^+, \mbox{ continuously non increasing density, positive on } [0,1-\delta) \\
& \qquad  \mbox{ and such that  } f_{|{[1-\delta, 1] }}=0  \} , \\
\mathcal{P}_{\delta}     &     =\big\{     \mathbb{P}_{\theta,     f};
\frac{d\mathbb{P}_{\theta, f}}{d\mu} 
= \theta +(1-\theta)f ; (\theta, f) \in (0,1)\times \mathcal{F}_{\delta}\big\} .
\end{align*}
Note  that for  any fixed  value $\delta  \in [0,  1)$,  the condition
stated  in  Proposition~\ref{prop:identif}  is  satisfied on  the  set
$\mathcal{F}_\delta$, namely forall  $f\in \mathcal{F}_\delta$ and for
all $c\in (0,1)$, we have $c+(1- c)f \notin \mathcal{F}_\delta$. Thus, the
parameter    $(\theta,f)$     is    identifiable    on    $(0,1)\times
\mathcal{F}_\delta$.

The  case $\delta  > 0$  corresponds to  models where  density  $f$ is
supposed to vanish on a set of points with non null Lebesgue
measure. This case is thus easier from an estimation perspective. Note
that when $\delta =0$, it is usual to add assumptions on $f$. Here, we
choose   to  consider   the  case   where   $f$  is   assumed  to   be
non increasing.  The same results may be obtained by replacing
  this assumption with a regularity  constraint on $f$. Note also that
  when $\delta  >0$, the assumption  that $f$ is non  increasing could
  be removed without any change in the results.

We  aim at  computing  the (asymptotic) efficient  information  for estimating  the
finite dimensional parameter $\psi(\mathbb{P}_{\theta,f}) = \theta$ in
model $\mathcal{P}_\delta$ where we consider $f$ as a nuisance parameter. 
We start  by recalling  some concepts  from  semiparametric theory and give explicit expressions of the objects arising from this theory in our specific framework. We follow the notation of Chapter 25 and more particularly Section 25.4 in~\cite{VanderVaart1998} and refer to this book for more details.\\

We fix a parameter value $(\theta,f)$ and  consider first a parametric submodel of $\mathcal{F}_\delta$ induced by the following path 
\begin{equation}
\label{eq-path}
t \mapsto f_t(x) = \frac{k(t h_0(x))f(x)}{\int k(t h_0(u))f(u)du} =
c(t)k(t h_0(x))f(x),
\end{equation}
where $h_0$ is a continuous and non increasing function on $[0, 1]$, the function $k$ is defined by  $k(u) = 2(1+e^{-2u})^{-1}$ and the normalising constant $c(t)$ satisfies $c(t)^{-1}=\int k(t h_0(u))f(u)du$. A tangent set ${}_{f}\mathcal{\dot{P}}_{\delta}$ is composed of the score functions associated to such parametric submodels (as $h_0$ varies). It is easy to see that the path~\eqref{eq-path} is differentiable and that its corresponding score function is obtained by differentiating $t\mapsto \log [\theta+(1-\theta)f_t(x)]$ at $t=0$. We thus obtain a tangent set for  $f$ given by
\begin{equation*}
{}_f\dot{\mathcal{P}}_{\delta} = \Big\{h= \frac{(1-\theta)fh_0}{\theta
  +(1-\theta)f}; h_0\ \mbox{is continuous and non increasing on } [0,1-\delta) \mbox{ with} \int fh_0 =0\Big\}.
\end{equation*}
Now, we consider parametric submodels of $\mathcal{P}_\delta$ induced by paths of the form $t\mapsto \mathbb{P}_{\theta+ta,f_t}$ where the paths $t\mapsto f_t$ in $\mathcal{F}_\delta$ are given by~\eqref{eq-path}.
We remark that  if  $\dot{l}_{\theta,f}$ is the ordinary score function
for $\theta$ in the model in which $f$ is fixed, then for every $a \in
\mathbb{R}$ and for every $h \in {}_f\dot{\mathcal{P}}_{\delta}$, we
have $a \dot{l}_{\theta,f} + h$ is a score function for $(\theta, f)$
corresponding to the path $t \mapsto
\mathbb{P}_{\theta+ta,f_t}$. Hence, a tangent set  $\dot{\mathcal{P}}_{\delta}$ of the model $\mathcal{P}_{\delta}$ at $\mathbb{P}_{\theta, f}$ with respect to the parameter $(\theta,f)$ is given by the linear span 
\[
\mathcal{\dot{P}_{\delta}}=\text{lin}\big( \dot{l}_{\theta, f}+ {}_f\mathcal{\dot{P}}_{\delta}\big)=\{\alpha\dot{l}_{\theta,   f} + \beta h ;  (\alpha,\beta) \in \mathbb{R}^2, h \in {}_f\mathcal{\dot{P}}_{\delta}\}.
\]
Moreover, the ordinary score function $\dot{l}_{\theta,f}$ for
$\theta$ in the model in which $f$ is fixed is given by 
\begin{equation}\label{eq:ordinary_score}
\dot{l}_{\theta,f}(x) = \frac{\partial}{\partial \theta}\log[\theta+(1-\theta)f(x)]=\frac{1-f(x)}{\theta
  +(1-\theta)f(x)}.
\end{equation}
Now we let $\tilde{l}_{\theta,f}$ be the efficient
score function and $\tilde{I}_{\theta,f}$ be the efficient information
for estimating $\psi (\mathbb{P}_{\theta,f})=\theta$. These quantities are defined respectively as
\[
\tilde{l}_{\theta,f}=
\dot{l}_{\theta,f}-\Pi_{\theta,f}\dot{l}_{\theta,f}\ \text{and }
\tilde{I}_{\theta,f} = \mathbb{P}_{\theta,f}(\tilde{l}_{\theta,f}^2),
\]
where $\Pi_{\theta,f}$ is the orthogonal projection onto the closure
of the linear span of ${}_f\mathcal{\dot{P}}_{\delta}$ in
$\mathbb{L}_2(\mathbb{P}_{\theta,f})$.
The functional $\psi: \mathbb{P}_{\theta,f}\mapsto \theta$ is said to be differentiable at $\mathbb{P}_{\theta,f}$ relative to the tangent set $\mathcal{\dot P}_\delta$ if there exists a continuous linear map $\tilde \psi_{\theta,f}:\mathbb{L}_2(\mathbb{P}_{\theta,f}) \mapsto \mathbb{R}$, called the efficient influence function, such that for every path $t\mapsto f_t$ with score function $h\in {}_f\mathcal{\dot P}_\delta$, we have
\[
\forall a \in \mathbb{R}, \quad a = \int  \tilde \psi_{\theta,f}(x) [a^\intercal\dot l_{\theta,f}(x)+h(x)] d\mathbb{P}_{\theta,f}(x).
\]
Setting $a=0$, we  see that this efficient influence  function must be
orthogonal           to            the           tangent           set
${}_{f}\mathcal{\dot{P}}_{\delta}$.  Finally,  note  that  under  some
assumptions, the efficient influence function $\tilde{\psi}_{\theta,f}$ equals 
$\tilde{I}_{\theta,f}^{-1}\tilde{l}_{\theta,f}$  \cite[see Lemma 25.25
in][]{VanderVaart1998}. 
The following theorem provides expressions for these quantities in our
setup. All the proofs in the current section are postponed to Section~\ref{sec:proof_lower}.

\begin{thm}\label{thm:I_and_psi}
The efficient score
function   $\tilde{l}_{\theta,f}$   and   the  efficient   information
$\tilde{I}_{\theta,f}$    for    estimating    $\theta$    in    model
$\mathcal{P}_\delta$ are given by 
\begin{equation}
  \label{eq:efficient_score}
\tilde{l}_{\theta,f}(x)                                               =
\frac{1}{\theta}-\frac{1}{\theta(1-\theta\delta)}\mathbf{1}_{[0, 1 -
  \delta)}(x) \quad \text{and}\quad \tilde{I}_{\theta,f} = \frac{\delta}{\theta(1-\theta\delta)} ,
\end{equation}
where $\mathbf{1}_A(\cdot)$ is the indicator function of set $A$.
In particular, when $\delta=0$, this efficient information is zero. In this case, the functional $\psi(\mathbb{P}_{\theta,f})=\theta$ is not differentiable at $\mathbb{P}_{\theta,f}$ relative to the tangent set  $\mathcal{\dot{P}}_{0}$. \\
Moreover, when $\delta >0$,  the efficient influence
function  $\tilde{\psi}_{\theta,f}$   relative  to  the   tangent  set
$\dot{\mathcal{P}}_{\delta}$  is given by 
\[
 \tilde{\psi}_{\theta,f}(x) = \frac{1}{\delta}\mathbf{1}_{[1-\delta, 1]}(x)-\theta.
\]
\end{thm}

This  theorem has  some  consequences  on the  quadratic  risk of  any
estimator that we now explain. 
For every score function $g$ in the tangent set $\mathcal{\dot{P}_{\delta}}$, we write $P_{t,g}$ for a path with score function $g$ along which the functional  $\psi: \mathbb{P}_{\theta,f}\mapsto \theta$ is differentiable. Namely, $P_{t,g}$ takes the form $\mathbb{P}_{\theta+ta,f_t}$ for some path  $t\mapsto f_t$ and some $a\in\mathbb{R}$. Now, an estimator sequence $\hat{\theta}_n$ is called regular at
$\mathbb{P}_{\theta,f}$ for estimating $\theta$
(relative to the tangent set $\mathcal{\dot{P}_{\delta}}$) if there
exists a probability measure $L$ such that for any  
score function $g\in \mathcal{\dot P}_\delta$ corresponding to a path of the form $t\mapsto (\theta+ta,f_t)$, we have
\[
\sqrt{n}\big(\hat{\theta}_n-\psi(P_{1/\sqrt{n},g})\big)= 
\sqrt{n} \Big[\hat{\theta}_n-\Big(\theta+\frac a {\sqrt{n}}\Big)\Big]
\xrightarrow[] {d} L, \mbox{ under } P_{1/\sqrt{n},g} , 
\]
where 
$\xrightarrow[]{d}$ denotes convergence in distribution.
According    to   a    convolution    theorem   \cite[see Theorem    25.20
in][]{VanderVaart1998},   this  limit   distribution  writes   as  the
convolution between some unknown distribution and the centered Gaussian distribution $N(0,\mathbb{P}_{\theta,f}   (\tilde{\psi}_{\theta,f}^2))$   with  variance
\[
\mathbb{P}_{\theta,f}      (\tilde{\psi}_{\theta,f}^2)      =     \int
\tilde{\psi}_{\theta,f}^2 d\mathbb{P}_{\theta, f} .
\]
Thus we  say that an estimator sequence is asymptotically efficient at $\mathbb{P}_{\theta,f}$ (relative to the tangent set
$\dot{\mathcal P}_{\delta}$) if it is regular at $\mathbb{P}_{\theta,f}$ with limit
distribution $L = N(0,\mathbb{P}_{\theta,f}( \tilde{\psi}_{\theta,f}^2))$, in other words it is the best regular
estimator. The quadratic risk of an estimator sequence $\hat{\theta}_n$ (relative to the tangent set $\dot{\mathcal P}_{\delta}$), is defined as 
\[
\underset{E_\delta}{\sup}\ \underset{n \rightarrow \infty}{\liminf}\ \underset{g
  \in E_\delta}{\sup}\ P_{1/\sqrt{n},g}
\big[\sqrt{n}\big(\hat{\theta}_n-\psi(P_{1/\sqrt{n},g})\big)\big]^2,
\]
where the first supremum is taken over all finite subsets $E_\delta$ of the
tangent set $\dot{\mathcal{P}}_{\delta}$. 
According to the local asymptotic minimax (LAM) theorem \cite[see
Theorem 25.21  in][]{VanderVaart1998}, this quantity  is lower bounded
by the minimal variance $\mathbb{P}_{\theta,f}
(\tilde{\psi}_{\theta,f}^2)$.  Thus,  Theorem~\ref{thm:I_and_psi}  has
the following corollary.

\begin{cor}\label{cor:LAM}
When $\delta =0$, any  estimator sequence ${\hat{\theta}_n}$ has an infinite
quadratic risk, namely 
\begin{equation*}
\underset{E_0}{\sup}\ \underset{n \rightarrow \infty}{\liminf}\ \underset{g
  \in E_0}{\sup}\ \esp _{P_{1/\sqrt{n},g}} \big[\sqrt{n}\big(\hat{\theta}_n-\psi(P_{1/\sqrt{n},g})\big)\big]^2 = +\infty,
\end{equation*}
where the first supremum is taken over all finite subsets $E_0$ of the
tangent set $\dot{\mathcal{P}}_{0}$.\\
When $\delta >0$, we obtain that  
\begin{itemize}
\item [i)] For any estimator sequence ${\hat{\theta}_n}$ we have, 
\begin{equation*}
\underset{E_\delta}{\sup}\ \underset{n \rightarrow \infty}{\liminf}\ \underset{g
  \in E_\delta}{\sup}\ \esp _{P_{1/\sqrt{n},g}}
\big[\sqrt{n}\big(\hat{\theta}_n-\psi(P_{1/\sqrt{n},g})\big)\big]^2 \geq
\theta\big(\frac{1}{\delta} - \theta\big),
\end{equation*}
where the first supremum is taken over all finite subsets $E_\delta$ of the
tangent set $\dot{\mathcal{P}}_{\delta}$.\\
\item [ii)] A sequence of estimators
$\hat{\theta}_n$ is asymptotically efficient in the sense of a 
convolution theorem (best regular estimator) if and only if it
satisfies
\begin{equation}\label{eq:efficiency}
\hat{\theta}_n =\frac{1}{n}\sum_{i=1}^n\frac{1}{\delta}\mathbf{1}_{X_i \in
  [1-\delta, 1]}+o_{\mathbb{P}_{\theta,f}}(n^{-1/2}).
\end{equation}
\end{itemize}
\end{cor}

\begin{remark}\label{rem:efficiency_oracle}
A) When $\delta=0$, using Theorem 2 in \cite{Chamberlain1986}, we conclude that there
is no regular estimator for $\theta$  relative to the tangent
set  $\dot{\mathcal{P}}_{0}$.  This implies  that  if  there exists  a
$\sqrt{n}$-consistent estimator  in model $\mathcal{P}_0$,  it can not
have finite asymptotic variance. 
In other words, we could have $\sqrt{n}(\hat \theta_n-\theta)=O_{\mathbb{P}}(1)$ for some estimator $\hat \theta_n$ but then $\Var(\sqrt{n}\hat \theta_n)\to+\infty$. 
However, we note that the only rates of convergence obtained until now
in this case are nonparametric ones.\\
B) When $\delta>0$,   for fixed parameter value $\lambda$  such that $G(\lambda)
  <1$,        \citeauthor{Storey2002}'s        estimator        $\hat
  \theta^{\text{Storey}}(\lambda)$ satisfies 
\[
\sqrt {n}\left( \hat \theta^{\text{Storey}}(\lambda) -\frac {1-G(\lambda)}{1-\lambda}
\right)   \xrightarrow[n\to  \infty]  {d}   N\left(0,  \frac
  {G(\lambda)(1-G(\lambda))}{(1-\lambda) ^2}
\right)
\]
\citep[see for instance][]{Genovese2004}.  In particular, if we assume
that  $f$   vanishes  on   $[\lambda  ,  1]$   then  we   obtain  that
$G(\lambda)=1-\theta(1-\lambda)$ and $\hat
  \theta^{\text{Storey}}(\lambda)$  becomes a  $\sqrt{n}$-consistent  estimator of
  $\theta$, which is moreover asymptotically distributed, with  asymptotic variance 
\[
\theta \left( \frac 1 {1-\lambda} -\theta\right).
\]
In  this  sense,   the  oracle  version  of  \citeauthor{Storey2002}'s
estimator  that picks $\lambda  =1-\delta$ (namely  choosing $\lambda$
as  the smallest  value such  that $f$  vanishes on  $[\lambda,1]$) is
asymptotically efficient. Note also that $\hat
  \theta^{\text{Storey}}(\lambda)$                        automatically
  satisfies~\eqref{eq:efficiency}.
\end{remark}

\section{Upper bounds for the quadratic risk and efficiency (when $\delta>0$)}\label{sec:sqrtn_consistency}
In this section, we  investigate the existence of asymptotically efficient estimators
for $\theta$, in the case where $\delta>0$. We consider histogram based estimators of $\theta$ where a nonparametric histogram estimator $\hat g$ of $g$ is combined with an interval selection that aims at picking an interval where $g$ is equal to $\theta$. We  start   by  establishing  the   existence  of  $\sqrt{n}$-consistent
estimators: a simple histogram based procedure is studied in Section~\ref{sec:his-estimator} while a more elaborate one is the object of Section~\ref{sec:CR-estimator}. Finally in Section~\ref{sec:one_step} we  explain the general one-step method  to construct an asymptotically      efficient     estimator     relying      on     a
$\sqrt{n}$-consistent procedure and  discuss conditions under which 
an asymptotically efficient estimator could be obtained in model $\mathcal{P}_\delta$.

\subsection{An histogram based estimator}\label{sec:his-estimator}
Throughout this  section and the following one, we assume that the density $f$ belongs to $\mathbb{L}^2([0,1])$.
Let  $\hat{g}_I$ be a histogram estimator corresponding to a partition $I=(I_k)_{1,\ldots,D}$ of
$[0,1]$, defined by
\begin{equation*}
\hat{g}_I(x)=\sum_{k=1}^D\frac{n_k}{n|
  I_k|}\mathbf{1}_{I_k}(x) , 
\end{equation*}
where $n_k=\text{card}\{i : X_i \in I_k\}$ is the number of
observations in $I_k$
and $|I_k|$ is the width of interval $I_k$. We estimate $\theta$
by the minimal value of $\hat{g}_I$, that is 
\begin{equation}
  \label{eq:his-estimator}
\hat   \theta_{I,n}   =    \min_{1\le   k\le   D}   \frac{n_k}{n|I_k|}
=\frac{n_{\hat k_n}}{n|I_{\hat k_n}|} ,
\end{equation}
where we let 
\[
\hat k_n \in \Argmin_{1\le k \le D} \left \{ \frac{ n_k} {n|I_k|} = \frac 1 {n
  |I_k|} \sum_{i=1}^n \mathbf{1}_{X_i\in I_k} \right \} .
\]

Note  that histogram estimators  are natural  nonparametric estimators
for $g$ when assuming that $f\in \mathcal{F}_\delta$ with $\delta >0$, that is $g$ is
constant on  an interval. It  is easy to  see that $\hat  \theta_{I,n}$ is
$\sqrt{n}$-consistent as soon as the partition $I$ is fine enough. We moreover establish that this estimator has a variance of the order $1/n$. The proof of this result appears in section~\ref{sec:proof_upper}.

\begin{prop}\label{prop:sqrtn_consistent}
 Fix $\delta >0$ and suppose that $f \in \mathcal{F}_\delta$. Assume
 moreover that the  partition $I$  is such  that
 $\max_k   |I_k|$   is  small   enough,   then the estimator  $\hat
 \theta_{I,n}$ has the following properties
\begin{itemize}
\item[i)] $\hat \theta_{I,n}$ converges almost surely to $\theta$,  
\item[ii)] $\hat \theta_{I,n}$ is $\sqrt{n}$-consistent, \emph{i.e.}
  $\sqrt{n} (\hat \theta_{I,n} -\theta) =O_{\mathbb{P}}(1)$, 
\item[iii)] $\underset{n \rightarrow \infty}{\limsup}
  \Var(\sqrt{n}\hat\theta_{I,n}) < +\infty$.
\end{itemize}
\end{prop}

Note that while $\sqrt{n}$-consistency and a control of the variance
of $\sqrt{n}\hat  \theta_{I,n}$ are  proved in the  above proposition,
asymptotic normality of $\hat \theta_{I,n}$ or the value of its
asymptotic variance are difficult to obtain. Indeed, for any deterministic interval $I_k$, the central limit theorem (CLT) applies on the estimator $n_k/(n|I_k|)$. However, an histogram based estimator such as $\hat \theta_{I,n}$ is based on the selection of a random interval $\hat I$ and the CLT fails to apply directly on $n_{\hat I}/(n|\hat I|)$. Note also that the choice of the partition $I$ is not solved here. From a practical point of view, decreasing the parameter $\max_k   |I_k|$ will in fact increase the variance of the estimator. In the next section, we study a procedure that automatically selects the best partition among a given collection.


\subsection{\cite{Celisse-Robin2010}'s procedure}\label{sec:CR-estimator}
We recall here the procedure for estimating $\theta$ that is
presented in \cite{Celisse-Robin2010}. It relies on  an elaborate histogram approach that selects the best partition among a given collection. As it will be seen from the simulations experiments (Section~\ref{sec:simus}), its asymptotic variance is likely to be smaller than for the previous estimator, justifying our interest into this procedure. Unfortunately, from a theoretical point of view, we only establish that this estimator should be as good as the previous one. Note that since not many estimators of $\theta$ have been proved to be $\sqrt{n}$-convergent, this is already a non trivial result.

For a given integer
$M$, define $\mathcal{I}_M$ as the set of partitions of $[0,1]$ such
that for some integers $k, l$ with $2 \leq k+2 \leq l \leq M$, the
first $k$ intervals and the last $M-l$ ones are regular of width
$1/M$, namely
\begin{equation*}
\mathcal{I}_M = \big\{ I=(I_i)_i : \forall i \neq k+1, |I_i| =
  \frac{1}{M},\ |I_{k+1}| = \frac{l-k}{M},\  2\leq k+2
\leq l \leq M \big\}.
\end{equation*}
These partitions are motivated by the assumption that $f$ vanishes on a set $[\lambda,\mu]\subset [0,1]$.  
 Then for two given integers $m_{min} <m_{max} $, denote by $\mathcal{I}$ the following collection of  partitions  
\begin{equation}
\label{eq-collection}
\mathcal{I} = \underset{m_{min} \leq m \leq m_{max}}{\bigcup} \mathcal{I}_{2^m}.
\end{equation}
Every partition $I$ in $\mathcal{I}$ is characterized by a triplet
$(M=2^m,\lambda=k/M,\mu=l/M)$ and the quality of the histogram estimator
$\hat{g}_I$ is measured by its quadratic risk.  So in this sense, the
\emph{oracle estimator} $\hat{g}_{I^\star}$ is obtained through 
\begin{equation*}
I^\star = 
 \underset{I\in
  \mathcal{I}}{\text{argmin}}\ \esp[||g-\hat{g}_I||_2^2]= \underset{I\in  \mathcal{I}}{\text{argmin}}\ R(I),
\text{ where }  R(I)=\esp\Big[||\hat{g}_I||_2^2 - 2\int_0^1\hat{g}_I(x)g(x)dx\Big].
\end{equation*} 
However, for every partition $I$, the quantity $R(I)$  depends on $g$ which is unknown. Thus $I^\star$ is  an oracle and not  an estimator.  It is
then  natural  to replace  $R(I)$  by an  estimator.  In
\cite{Celisse-Robin2008, Celisse-Robin2010}, the authors use 
leave-p-out (LPO) estimator of $R(I)$ with
$p\in\{1,\ldots,n-1\}$, whose expression is given by \cite[see][Theorem 2.1]{Celisse-Robin2008}
\begin{equation}
\label{LPO-estimator}
\hat{R}_p(I) =
\frac{2n-p}{(n-1)(n-p)}\displaystyle\sum_k\frac{n_k}{n|I_k|} -  \frac{n(n-p+1)}{(n-1)(n-p)}\displaystyle\sum_k\frac{1}{|I_k|}\big(\frac{n_k}{n}\big)^2.
\end{equation}
The best theoretical value of $p$ is the one that minimizes the mean squared  error
(MSE) of $\hat{R}_p(I)$, namely 
\begin{equation*}
p^\star(I) = \underset{p\in
  \{1,\ldots,n-1\}}{\text{argmin}}MSE(p,I)=\underset{p\in
  \{1,\ldots,n-1\}}{\text{argmin}}\esp\Big[\big(\hat{R}_p(I) -
R(I)\big)^2\Big].
\end{equation*}
It  clearly  appears  that  $MSE(p,I)$  has the  form  of  a  function
$\Phi(p,I,\alpha)$    \cite[see][Proposition   2.1]{Celisse-Robin2008}
depending on the  unknown vector $\alpha = (\alpha_1,  \alpha_2, \ldots ,
\alpha_D)$ with $\alpha_k = \mathbb{P}(X_1\in I_k)$. A natural idea is then  to replace the
$\alpha_k$s in $\Phi(p,I,\alpha)$ by their empirical counterparts $\hat{\alpha}_k=n_k/n$ and
an estimator of $p^\star(I)$ is therefore given by
\begin{equation*}
\hat{p}(I) = \underset{p\in
  \{1,\ldots,n-1\}}{\text{argmin}}\widehat{MSE}(p,I) = \underset{p\in
  \{1,\ldots,n-1\}}{\text{argmin}}\Phi(p,I,\hat{\alpha}).
\end{equation*}
The exact calculation of $\hat{p}(I)$ may be found in Theorem~3.1 from
\cite{Celisse-Robin2008}. Hence, the procedure for estimating $\theta$ is the following one
\begin{enumerate}
\item For each partition $I \in \mathcal{I}$, define
  $\hat{p}(I) = \underset{p\in
  \{1,\ldots,n-1\}}{\text{argmin}}\widehat{MSE}(p,I)$,
\item Choose $\hat{I} =
  (\hat{M},\hat{\lambda},\hat{\mu}) \in \underset{I \in
    \mathcal{I}}{\text{argmin}}\ \hat{R}_{\hat{p}(I)}(I)$ such that
  the width of the interval $[\hat{\lambda},\hat{\mu}]$ is maximum,
\item Estimate $\theta$ by $\hat{\theta}^{CR}_n = \text{card}\{i:X_i \in
  [\hat{\lambda},\hat{\mu}]\}/[n(\hat{\mu} - \hat{\lambda})]$.
\end{enumerate}
\begin{remark}
In our procedure, we consider the set of natural partitions defined by~\eqref{eq-collection}, while \cite{Celisse-Robin2010} use the one
defined by 
\begin{equation*}
\mathcal{I} = \underset{M_{min} \leq M \leq M_{max}}{\bigcup} \mathcal{I}_{M}.
\end{equation*}
This change is natural for lowering the complexity of the algorithm
and has no consequences on the theoretical properties of the estimator. In particular, if we assume the function $f$
vanishes on an interval $[1-\delta, 1]$, then the complexity of the
algorithm is simpler when we consider the following set of
partitions
\begin{equation*}
\mathcal{I} = \underset{m_{min} \leq m \leq m_{max}}{\bigcup} \mathcal{I}_{2^m},
\end{equation*} 
where
\begin{equation*}
\mathcal{I}_M = \big\{ I^{(k)}=(I_i)_{i=1,\ldots,k+1} : \forall i \leq k, |I_i| =
  \frac{1}{M},\ |I_{k+1}| = \frac{M-k}{M},\  1\leq k \leq M-2 \big\}.
\end{equation*}
\end{remark}

In \cite{Celisse-Robin2010}, the authors only establish convergence in probability of this estimator. Here, we prove its almost sure convergence, $\sqrt{n}$-consistency and establish that its variance is of the order $1/n$. Let us first introduce some assumptions. 

\begin{assum}
\label{assum1}
Density $f$ is null  on an interval $[\lambda^\star, \mu^\star]\subset
(0,1]$ (with  unknown values $\lambda^\star$ and  $\mu^\star$) and $f$
is monotone outside the interval $[\lambda^\star,\mu^\star]$.
\end{assum}

For example, $f$ is decreasing on
$[0,\lambda^\star]$    and   increasing   on    $[\mu^\star,1]$.   This assumption is
stronger  than  Assumption  A'  in \cite{Celisse-Robin2010},  the
latter  not  being  sufficient  to  establish    the  result  they
  claim (see  the proof of Lemma~\ref{lem3} for more  details).  The
monotonicity part of our assumption is not
necessary and we shall explain what is exactly required and how we use
the previous assumption in the proof of Lemma~\ref{lem3}. Under Assumption~\ref{assum1}, the true parameter $\theta$ is equal to $g(x)$ for all $x$
in $[\lambda^\star,  \mu^\star]$. Note that  the case where  we impose
$\mu^\star=1$ is included in this setting. 
We now introduce a technical condition that comes from~\cite{Celisse-Robin2010}. We let 
\[
 \forall (i,j) \in \mathbb{N}^2, \quad 
 s_{ij} =  \sum_{k=1}^D \frac{\alpha_k^i}{|I_k|^j} ,
\]
and further assume that the collection of partitions $\mathcal{I}$ and density $f$ are such that 
\begin{equation}
\label{condition-sij}
\forall I \in \mathcal{I}, \quad 
8s_{11}s_{21}-2s_{11}^2+8s_{32}-10s_{21}^2-4s_{22} \neq 0,\ s_{21}-s_{22}-s_{32}+3s_{11} \neq 0.
\end{equation}
This technical condition is used in~\cite{Celisse-Robin2010} to control the behaviour of the minimizer $\hat p(I)$.  We are now ready to state our result, whose proof can be found in Section~\ref{sec:proof_CR}.

\begin{prop}\label{prop:CR_sqrtn_consistent}
 Suppose that $f$ satisfies Assumption~\ref{assum1} as well as the technical condition~\eqref{condition-sij}. Assume
 moreover that  $m_{\max}$   is  large   enough,   then the estimator  $\hat
 \theta_{n}^{CR}$ has the following properties
\begin{itemize}
\item[i)] $\hat \theta^{CR}_{n}$ converges almost surely to $\theta$,  
\item[ii)] $\hat \theta^{CR}_{n}$ is $\sqrt{n}$-consistent, \emph{i.e.}
  $\sqrt{n} (\hat \theta^{CR}_{n} -\theta) =O_{\mathbb{P}}(1)$, 
\item[iii)] If $p$ is fixed then $\underset{n \rightarrow \infty}{\limsup}
  \Var(\sqrt{n}\hat\theta^{CR}_{n}) < +\infty$. 
\end{itemize}
\end{prop}

Here again, asymptotic normality of $\hat \theta^{CR}_{n}$ or the exact value of its
asymptotic variance are difficult to obtain. Heuristically, one can explain that this procedure outperforms the simpler histogram based with fixed partition approach described in the previous section. Indeed, when considering a fixed partition, the latter should be fine enough to obtain convergence but refining the partition increases the variance of $\hat \theta_{I,n}$. Here, \citeauthor{Celisse-Robin2010}'s approach realizes a compromise on the size of the partition that is used.

\subsection{One-step estimators}\label{sec:one_step}
In this section, we introduce the one-step method to
construct an asymptotically efficient estimator, relying on a $\sqrt{n}$-consistent
one~\cite[see][Section 25.8]{VanderVaart1998}. Let $\hat{\theta}_n$
be a $\sqrt{n}$-consistent estimator of $\theta$, then
$\hat{\theta}_n$ can be discretized on grids of mesh
width $n^{-1/2}$. Suppose that we are given a sequence of estimators
$\hat{l}_{n,\theta} (\cdot)= \hat{l}_{n,\theta}(\cdot;X_1,\ldots,X_n)$ of
the efficient score function $\tilde{l}_{\theta,f}$. Define
with $m=\lfloor n/2 \rfloor$,
\[
  \hat{l}_{n,\theta,i}(\cdot) = \left\{
          \begin{array}{ll}
             \hat{l}_{m,\theta}(\cdot;X_1,\ldots,X_m)& \qquad \mathrm{if}\quad i> m, \\
             \hat{l}_{n-m,\theta}(\cdot;X_{m+1},\ldots,X_n) & \qquad \mathrm{if} \quad i\leq m.\\
          \end{array}
        \right.
\]
Thus, for $X_i$ ranging through each of the two halves of the sample,
we use an estimator $\hat{l}_{n,\theta,i}$ based on the other half of
the sample. We assume that, for every deterministic sequence $\theta_n = \theta +
O(n^{-1/2})$, we have
\begin{eqnarray}
\label{equa-vitesse-para}
\sqrt{n}\mathbb{P}_{\theta_n,f}\hat{l}_{n,\theta_n}\xrightarrow[n\to\infty]
{\mathbb{P}_{\theta,f}} 0, \\
\label{equa-variance}
\mathbb{P}_{\theta_n,f}\Arrowvert \hat{l}_{n,\theta_n}
-\tilde{l}_{\theta_n,f}\Arrowvert ^2\xrightarrow[n\to\infty]
{\mathbb{P}_{\theta,f}} 0, \\
\label{equa-int}
\int \Arrowvert \tilde{l}_{\theta_n,f}d\mathbb{P}_{\theta_n,f}^{1/2}
-\tilde{l}_{\theta,f}d\mathbb{P}_{\theta,f}^{1/2}\Arrowvert ^2\xrightarrow[n\to\infty]
{} 0.
\end{eqnarray}
Note that in the above notation, the term $\mathbb{P}_{\theta_n,f} \hat l$
for some random function $\hat  l$ is an abbreviation for the integral
$\int \hat l(x) d \mathbb{P}_{\theta_n,f}(x)$. 
Thus the  expectation is taken  with respect to  $x$ only and  not the
random variables in $\hat l$.
Now under the above assumptions, the one-step estimator defined as
\begin{equation*}
\tilde{\theta}_n = \hat{\theta}_n - \Big( \displaystyle\sum_{i=1}^n\hat{l}_{n,\hat{\theta}_n,i}^2(X_i) \Big)^{-1}\displaystyle\sum_{i=1}^n\hat{l}_{n,\hat{\theta}_n,i}(X_i),
\end{equation*}
is asymptotically efficient at $(\theta, f)$ \cite[see][Section 25.8]{VanderVaart1998}. This estimator
$\tilde{\theta}_n$ can be considered a one-step iteration of the
Newton-Raphson algorithm for solving an approximation of the equation
$\sum_{i} \tilde{l}_{\theta,f}(X_i)=0$ with respect to $\theta$, starting at
the initial guess $\hat{\theta}_n$.

Now,  we  discuss  a  converse result  on  necessary
conditions for  existence of an asymptotically  efficient estimator of
$\theta$ and its implications in model $\mathcal{P}_\delta$.


Under condition~\eqref{equa-int}, it is shown in Theorem
7.4 from \cite{VanderVaartlecture} that the existence of
an asymptotically efficient sequence of estimators of $\theta$ implies
the existence of a sequence of estimators $\hat{l}_{n,\theta}$ of $\tilde{l}_{\theta,f}$ satisfying~\eqref{equa-vitesse-para} and~\eqref{equa-variance}. In our
case, it is not difficult to prove that condition~\eqref{equa-int}
holds. Then, the estimator $\hat{l}_{n,\theta}$ of the efficient score function $\tilde{l}_{\theta,f}$ must
satisfy both a "no-bias"~\eqref{equa-vitesse-para} and a consistency~\eqref{equa-variance}
condition. The consistency is usually easy to arrange, but the
"no-bias" condition requires a convergence to zero of the bias
at a rate faster than $1/\sqrt{n}$. We thus obtain the following proposition, whose proof can be found in Section~\ref{sec:proof_upper}.

\begin{prop}\label{prop:equiv_efficiency}
  The existence of an asymptotically efficient sequence of
estimators of $\theta$ in model $\mathcal{P}_\delta$ is equivalent to
the existence of a sequence of estimators $\hat{l}_{n,\theta}$  of
the efficient score function $\tilde{l}_{\theta,f}$ satisfying~\eqref{equa-vitesse-para} and~\eqref{equa-variance}. 
Moreover,   if   the   efficient   score   function
$\tilde{l}_{\theta,f}$  is  estimated through  a  plug-in method  that
relies on an estimate $\hat{\delta}_n$ of the parameter $\delta$, then
this condition is equivalent to $\sqrt{n}(\hat{\delta}_n- \delta) =
o_{\mathbb{P}}(1)$. 
\end{prop}

Let us  now explain the  consequences of this result.  The proposition
states  that efficient  estimators of  $\theta$ exist  if and  only if
estimators           of           $\tilde{l}_{\theta,f}$          that
satisfy~\eqref{equa-vitesse-para} and~\eqref{equa-variance} can be 
constructed. As  there is no  general method to estimate  an efficient
score function, such an estimator should rely on the specific 
expression~\eqref{eq:efficient_score}.   Though   we   cannot   claim  that   all   estimators   of
$\tilde{l}_{\theta,f}$ are plug-in estimates  based on an estimator of
the       parameter       $\hat       \delta$       plugged       into
expression~\eqref{eq:efficient_score}, it is likely to be the case. 
Then, existence  of efficient estimators of $\theta$  is equivalent to
existence  of estimators  of  $\delta$ that  converge  at faster  than
parametric rate. Note that this is possible for irregular models \citep[see Chapter 6 in][for more details]{Ibra}. However, for regular models, such estimators cannot be constructed and one might conjecture that efficient estimators of $\theta$ do not exist in regular models.


\section{Simulations}\label{sec:simus}
In this section, we give some illustrations of the previous results on
some simulated experiments  and explore the non asymptotic
performances of the estimators of $\theta$ previously discussed. We choose to
compare three different estimators: the  histogram based estimator $\hat \theta_{I,n}$ defined in Section~\ref{sec:his-estimator} through~\eqref{eq:his-estimator}, the more elaborate histogram based estimator $\hat \theta^{CR}_n$ proposed in~\cite{Celisse-Robin2010} and finally~\cite{Langaas2005}'s estimator, denoted by $\hat \theta^{L}_n$ and defined as the value $\hat g(X_{(n)})$ where $X_{(n)}$ is the largest $p$-value and $\hat g$ is Grenander's estimator of a decreasing density. We investigate the behaviour of these three different
estimators of $\theta$ under two different setups: $\delta =0$ and
$\delta \in (0, 1)$. More precisely, we consider  the alternative
density $f$ given by 
\begin{equation*}
f(x)=\frac{s}{1-\delta}\Big(1-\frac{x}{1-\delta}\Big)^{s-1}\mathbf{1}_{[0,1-\delta]}(x),
\end{equation*}
 where $\delta\in [0,1)$ and $s>1$. This form of density is 
 introduced in~\cite{Celisse-Robin2010} and covers various situations when varying its parameters. Note that $f$ is
 always decreasing, convex when $s\ge 2$ and concave when $s\in
 (1,2]$.  In the experiments, we consider a total of $8$ different models corresponding to different parameter values. These models are labeled as described in Table~\ref{tab:parameters}, distinguishing the cases $\delta=0$ and $\delta >0$. 
As an illustration, we represent some of the densities obtained for the
 $p$-values corresponding to $4$ out of the $8$ models in Figure~\ref{fig1:illustration}. For each estimator $\hat{\theta}_n$ of $\theta$, we compare the quantity  $ n\esp[(\hat{\theta}_n-\theta)^2]$ with the optimal
variance $\theta(\delta^{-1} -\theta)$ when this bound
exists. Equivalently, we compare the logarithm of mean squared error, $\log(\mbox{MSE})= \log \esp[(\hat{\theta}_n-\theta)^2]$ for each estimator $\hat{\theta}_n$ with $-\log(n)+\log[\theta(\delta^{-1} -\theta)]$. When $\delta =0$, we only
compare the slope of the line induced by $\log(\mbox{MSE})$ with the
parametric rate corresponding to a slope $-1$. In each case, we simulated data with sample size
$n\in \{5000; 7000; 9000; 10000; 12000; 14000;15000\}$ and perform
$R=100$ repetitions.

When computing the estimator $\hat\theta_{I,n}$, the choice of the partition $I$ surely affects the results. Here, we have chosen a regular partition $I$  such that it  is fine enough (we fixed $|I_k|<\delta$) but not too fine (choosing a too small value of $|I_k|$ increases the variance). The choice of the partition in the simple procedure $\hat \theta_{I,n}$ is an issue for real data problems. Our  goal here is to show that on simulated experiments, the "best" of these estimators still has a larger variance than $\hat \theta^{CR}_n$. Note that the partition $I$ is always included in the collection $\mathcal{I}$ of partitions from which $\hat \theta^{CR}_n$ is computed.

\begin{figure}[htbp]
\begin{center}
\includegraphics[width=0.6\columnwidth]{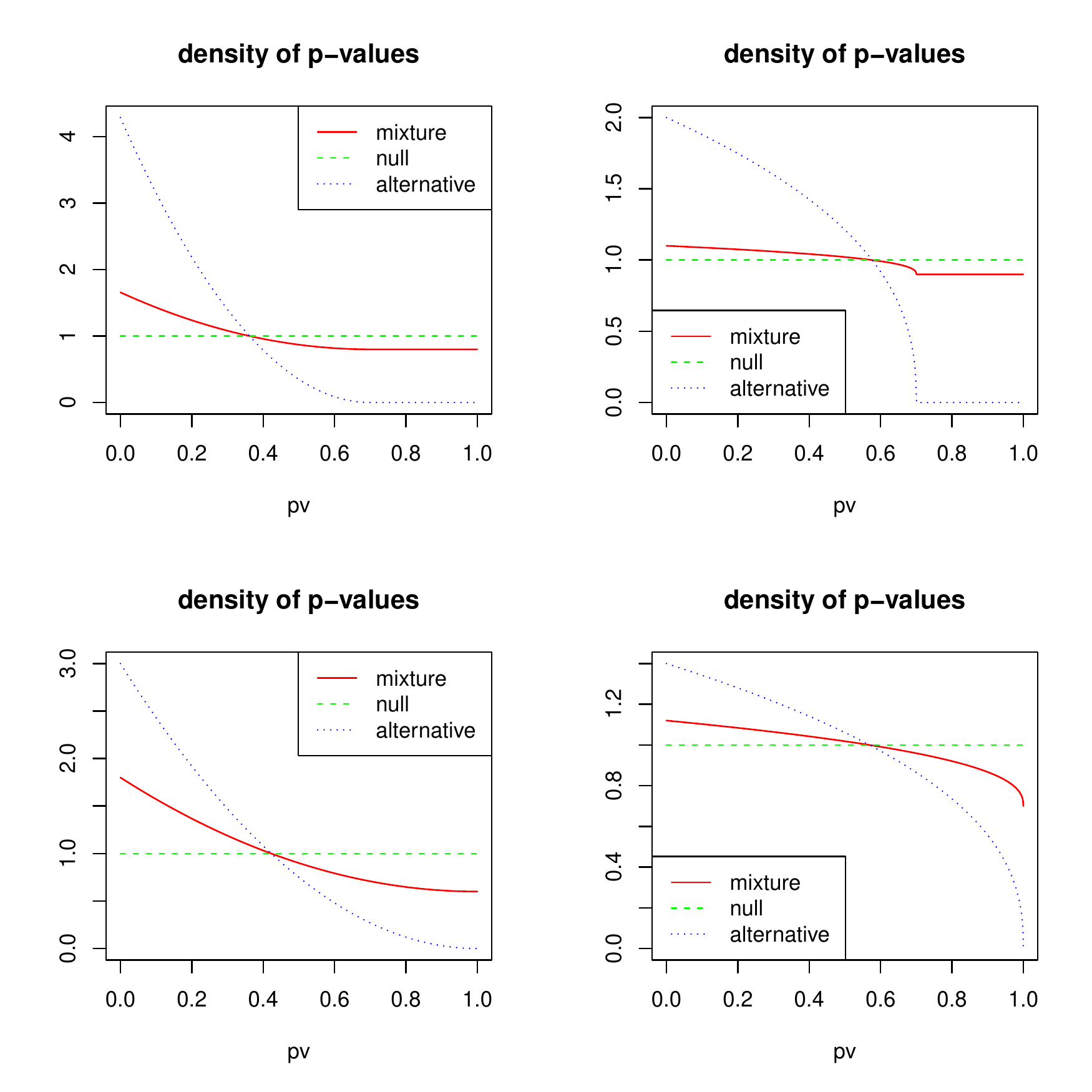}
\caption{Density function of the $p$-values. Top left: model ($b_1$); top right: model ($d_1$); bottom left: model ($a_2$); bottom right: model ($c_2$).}
\label{fig1:illustration}
\end{center}
\end{figure}

\begin{table}[h]
  \centering
  \begin{tabular}{c|c|c}
  $(s,\theta)$& $\delta = 0.3$ &$\delta = 0$\\
\hline 
$(3,0.6)$& $(a_1)$ & $(a_2)$ \\
$(3,0.8)$& $(b_1)$ & $(b_2)$ \\
$(1.4,0.7)$& $(c_1)$ & $(c_2)$ \\
$(1.4,0.9)$& $(d_1)$ & $(d_2)$ 
  \end{tabular}
  \caption{Labels of the 8 models with different parameter values.}
  \label{tab:parameters}
\end{table}

The results are presented in Figure~\ref{fig2:illustration} for the
case $\delta > 0$ and Figure~\ref{fig3:illustration} for the case $\delta =
0$. First, we note that in both cases ($\delta > 0$ and $\delta = 0$), \citeauthor{Langaas2005}'s estimator $\hat \theta^{L}_n$ has nonparametric rate of convergence (null slope) and performs badly compared to  $\hat \theta_{I,n}$ and $\hat \theta^{CR}_{n}$. In particular, when $\delta =0$ the two histogram based procedures  $\hat \theta_{I,n}$ and $\hat \theta^{CR}_{n}$ have better performances than the estimator $\hat \theta^{L}_n$ despite the fact that the latter is dedicated to the convex decreasing setup. Now, when $\delta >0$, both estimators $\hat \theta_{I,n}$ and $\hat \theta^{CR}_{n}$ exhibit a parametric rate of convergence (slope equal to $-1$). Moreover,  $\hat \theta^{CR}_{n}$ has a smaller variance than $\hat \theta_{I,n}$ (smaller intercept) and this variance is very close to the optimal one $\theta(\delta^{-1}-\theta)$. 
Now, when $\delta =0$, we observe two different behaviors depending on whether $f$ is convex or not. Indeed, for models ($a_2$) and ($b_2$) corresponding to the convex case, we observe that both estimators $\hat \theta_{I,n}$ and $\hat \theta^{CR}_{n}$ still exhibit a parametric rate of convergence, with a smaller variance for $\hat \theta^{CR}_{n}$. These estimators are thus robust to the assumption that $f$ vanishes on an interval in the convex setup. The results are slightly different when considering models ($c_2$) and ($d_2$) where $f$ is now concave. These estimators have a more erratic behaviour, exhibiting either parametric rate of convergence ($\hat \theta^{CR}_{n}$ in model ($c_2$) and $\hat \theta_{I,n}$ in model ($d_2)$) or nonparametric rates. Their respective performances in terms of variance are also less clear.
Nonetheless we conclude that  $\hat \theta^{CR}_{n}$ seems to exhibit the overall best performances, with parametric rate of convergence and almost optimal asymptotic variance.

\begin{figure}[htbp]
\begin{center}
\includegraphics[width=0.6\columnwidth]{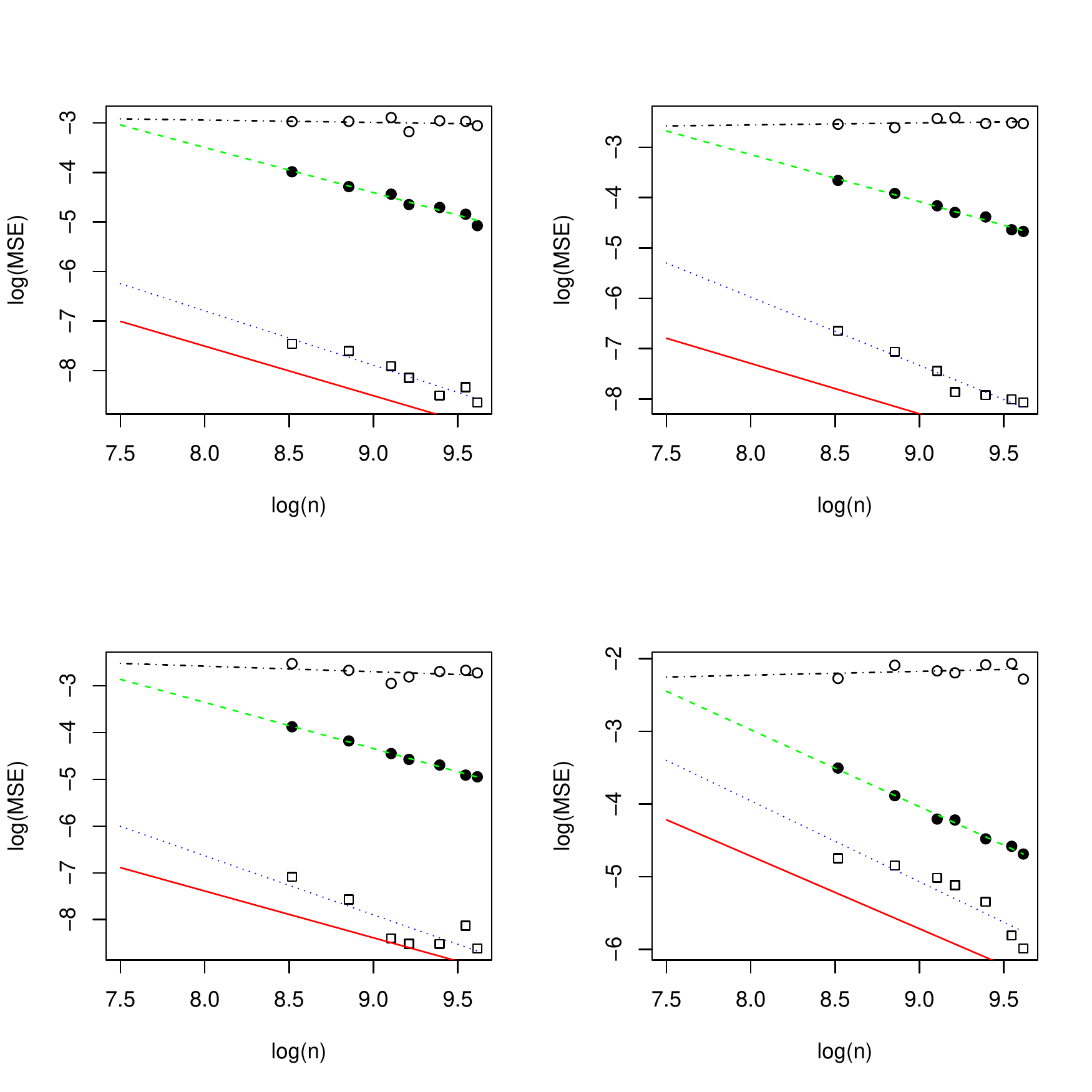}
\caption{Logarithm of the mean squared error  as a function of
  $\log(n)$ and corresponding linear regression for
 $\hat \theta^{L}_n$  ($\circ$ and black line,
  respectively), $\hat \theta^{CR}_n$ ($\square$ and
  blue line, respectively)  and $\hat \theta_{I,n}$ ($\bullet$ and green line, respectively) in  the case $\delta = 0.3$, for different parameter values: $(a_1)$ top left; $(b_1)$ top right; $(c_1)$ bottom left; $(d_1)$ bottom right. Red line represents the line   $y=-\log(n)+\log[\theta(\delta^{-1}-\theta)]$. }
\label{fig2:illustration}
\end{center}
\end{figure}

\begin{figure}[htbp]
\begin{center}
\includegraphics[width=0.6\columnwidth]{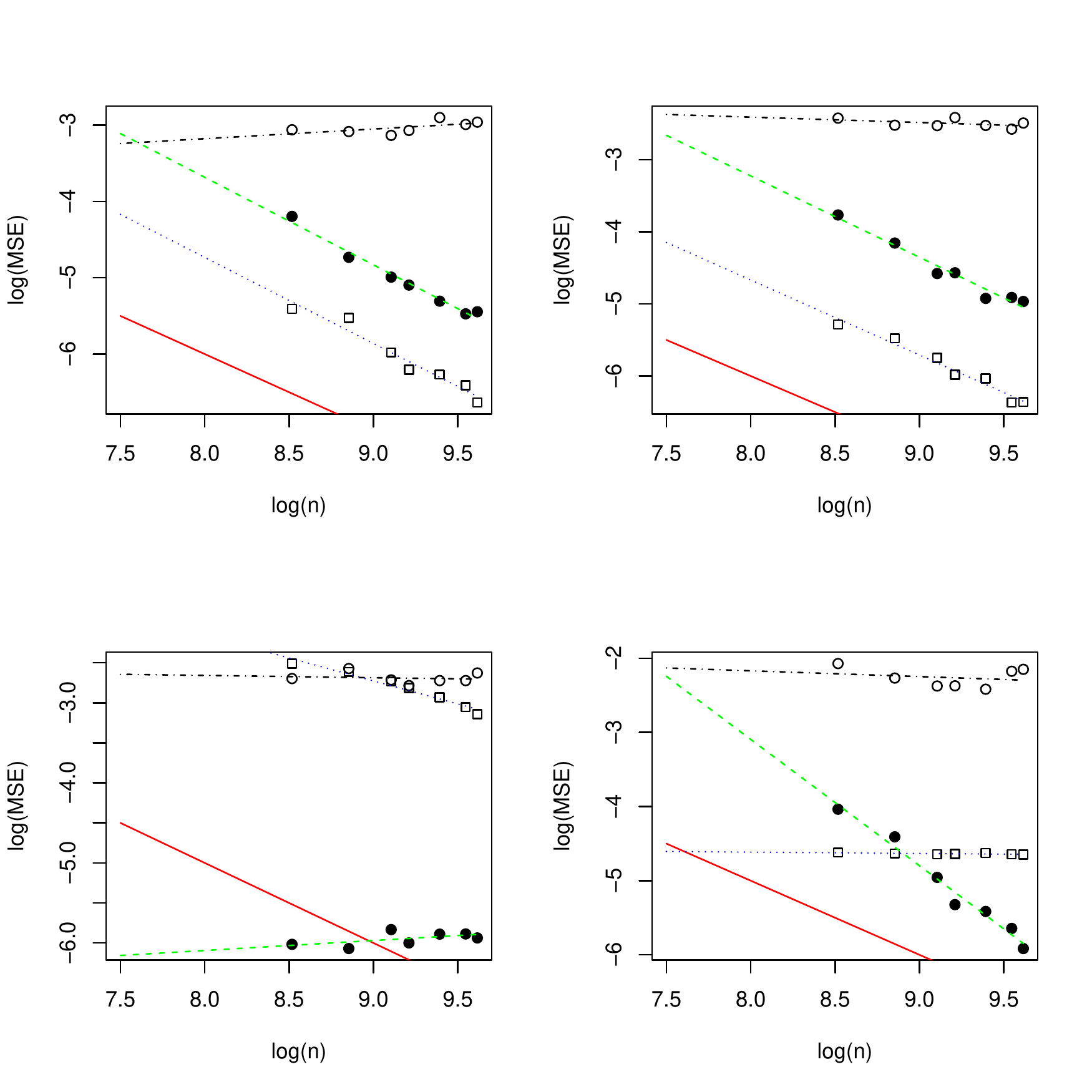}
\caption{Logarithm of the mean squared error  as a function of $\log(n)$ and corresponding linear regression for  $\hat \theta^{L}_n$ ($\circ$ and black line,
  respectively),  $\hat \theta^{CR}_n$ ($\square$ and
  blue line, respectively)  and  $\hat \theta_{I,n}$ ($\bullet$ and green line, respectively) in   the case $\delta = 0$, for different parameter values: $(a_2)$ top
  left; $(b_2)$ top right; $(c_2)$ bottom left; $(d_2)$ bottom right. Red
  line represents the line $y=-\log(n)+c$ for some well chosen constant $c$.}
\label{fig3:illustration}
\end{center}
\end{figure}

\section{Proofs}\label{sec:proofs}

\subsection{Proofs from Introduction and Section~\ref{sec:efficiency}}\label{sec:proof_lower}

\begin{proof}[Proof of Proposition~\ref{prop:identif}]
Sufficiency: Let us suppose that for all $f \in \mathcal{F}$ and for all
$c\in (0,1)$, we have $c+(1-c)f \notin \mathcal{F}$. We prove that the
parameters $\theta$ and $f$ are identifiable on the set $(0,1)\times
\mathcal{F}$ by contradiction. Suppose that there exist $(\theta_1,
f_1)$ and $(\theta_2,f_2) \in \mathcal{F}$, $(\theta_1,
f_1) \neq (\theta_2,f_2)$ such that 
\begin{equation}
\label{identifiable_1}
\theta_1 +(1-\theta_1) f_1(x) = \theta_2 +(1-\theta_2) f_2(x),\
\text{for all}\ x \in [0,1].  
\end{equation}
We can always consider $\theta_1 > \theta_2$. Let us denote by
$c=(\theta_1-\theta_2)/(1-\theta_2)$, then $c \in (0,1)$. We obtain
that 
\begin{equation}
\label{identifiable_2}
\theta_1 +(1-\theta_1) f_1(x) = \theta_2 +(1-\theta_2)(c+(1-c)f_1(x)),\
\text{for all}\ x \in [0,1].
\end{equation}
From \eqref{identifiable_1} and \eqref{identifiable_2}, we have $f_2
=c+(1-c)f_1$, it means that there exist $f_1 \in \mathcal{F}$ and $c
\in (0,1)$ such that $c+(1-c)f_1 \in \mathcal{F}$. So we have a contradiction.\\
Necessity: Suppose that the parameters $\theta$ and $f$ are
identifiable on the set $(0,1) \times \mathcal{F}$. We prove by
contradiction that for all
$f \in \mathcal{F}$ and for all $c \in (0,1)$, we have  $c+(1-c)f
\notin \mathcal{F}$. Indeed, suppose that there exist $f \in
\mathcal{F}$ and $c \in (0,1)$ such that $c+(1-c)f
\in \mathcal{F}$. For all $\theta_1 \in (0,1)$, we denote
$\theta_2=c+(1-c)\theta_1$, then we obtain
\begin{equation*}
\theta_1 +(1-\theta_1) (c+(1-c)f(x)) = \theta_2 +(1-\theta_2)f(x),\
\text{for all}\ x \in [0,1].
\end{equation*}
This implies that $\theta$ and $f$ are not
identifiable on the set $(0,1) \times \mathcal{F}$.
\end{proof}

\begin{proof}[Proof of Theorem~\ref{thm:I_and_psi}]
According to the  expression~\eqref{eq:ordinary_score} of the ordinary
score, we can write 
\begin{equation}\label{eq:dotl_decompose}
\dot{l}_{\theta,f} (x)=\Big(\frac{1-f}{\theta
  +(1-\theta)f}+\frac{\delta}{1-\theta\delta}\Big)\mathbf{1}_{[0,
  1-\delta)}(x)+\frac{1}{\theta}\mathbf{1}_{[1-\delta,
  1]}(x)-\frac{\delta}{1-\theta\delta}\mathbf{1}_{[0, 1-\delta)}(x). 
\end{equation}
Let us recall that  $\Pi_{\theta,f}$ is the orthogonal projection onto
the closure of the linear span of ${}_f\mathcal{\dot{P}}_{\delta}$ in $\mathbb{L}_2(\mathbb{P}_{\theta,f})$.
We prove  that the orthogonal  projection of $\dot  l_{\theta,f}$ onto
this space is equal to the first term appearing in the right-hand side
of~\eqref{eq:dotl_decompose}, namely
\begin{equation}
\label{equa_projection}
\Pi_{\theta,f}\dot{l}_{\theta,f}(x) =\Big(\frac{1-f(x)}{\theta
  +(1-\theta)f(x)}+\frac{\delta}{1-\theta\delta}\Big)\mathbf{1}_{[0,1 -\delta)}(x),
\end{equation}
and then the efficient score function for $\theta$ is 
\begin{equation*}
\tilde{l}_{\theta,f}(x) =
\dot{l}_{\theta,f}(x)-\Pi_{\theta,f}\dot{l}_{\theta,f}(x)=\frac{1}{\theta}\mathbf{1}_{[1-\delta,
  1]}(x)-\frac{\delta}{1-\theta\delta}\mathbf{1}_{[0, 1-\delta)}(x).
\end{equation*}
In fact, we can write
\begin{equation*}
-\Big(\frac{1-f}{\theta
  +(1-\theta)f}+\frac{\delta}{1-\theta\delta}\Big)\mathbf{1}_{[0, 1-\delta)}=\frac{(1-\theta)fh_0}{\theta
  +(1-\theta)f},
\end{equation*}
where
\begin{align*}
h_0(x)&=-\Big(\frac{1-f(x)}{(1-\theta)f(x)}+\frac{\delta}{1-\theta\delta}\times\frac{\theta+(1-\theta)f(x)}{(1-\theta)f(x)}\Big)\mathbf{1}_{[0,
1-\delta)}(x)\\
&= \frac{1}{(1-\theta)(1-\theta \delta)}\Big(1-\delta-\frac{1}{f(x)}\Big) \mathbf{1}_{[0,
1-\delta)}(x).
\end{align*}
The function $h_0$ is continuous and decreasing on $[0, 1-\delta)$. It is not
difficult to examine the condition $\int fh_0 = 0$. Indeed, 
\begin{align*}
\int_0^1f(x)h_0(x)dx&= \frac{1}{(1-\theta)(1-\theta
  \delta)}\int_0^{1-\delta}[(1-\delta)f(x)-1]dx\\
&= \frac{1}{(1-\theta)(1-\theta
  \delta)}\Big[\int_0^1(1-\delta)f(x)dx-(1-\delta)\Big] =0.
\end{align*}
Hence
\[
\Big(\frac{1-f}{\theta
  +(1-\theta)f}+\frac{\delta}{1-\theta\delta}\Big)\mathbf{1}_{[0,
  1-\delta)}\ \text{belongs to } \overline{\text{lin} ({}_f\mathcal{\dot{P}}_{\delta})}.
\]
Now, to conclude the proof of~\eqref{equa_projection}, it is necessary to
establish   that   the   second   term   in  the   right   hand   side
of~\eqref{eq:dotl_decompose}  is  orthogonal  to  the closure  of  the
linear span of ${}_f\mathcal{\dot{P}}_{\delta}$, namely 
\[
\frac{1}{\theta}\mathbf{1}_{[1-\delta,1]}-\frac{\delta}{1-\theta\delta}\mathbf{1}_{[0,
  1- \delta)}=\frac{1}{\theta(1-\theta\delta)}\mathbf{1}_{[0,1-\delta)}-\frac{\delta}{1-\theta\delta}\
\bot\ \overline{\text{lin} ({}_f\mathcal{\dot{P}}_{\delta})} , 
\]
where $\bot$ means orthogonality in $\mathbb{L}^2(\mathbb{P}_{\theta,
  f})$. In fact, for every score function 
\[
h =
\frac{(1-\theta)f             h_0}{\theta+(1-\theta)f}             \in
{}_f\mathcal{\dot{P}}_{\delta},
\]
the scalar  product between $h$ and the  remaining term in~\eqref{eq:dotl_decompose}
is given by
\begin{align*}
 &\int_0^1\big[\frac{1}{\theta(1-\theta\delta)}\mathbf{1}_{[0, 1- \delta)}(x)-\frac{\delta}{1-\theta\delta}\big]h(x)d\mathbb{P}_{\theta,f}(x)\\
&=\int_0^1\big[\frac{1}{\theta(1-\theta\delta)}\mathbf{1}_{[0, 1- \delta)}(x)-\frac{\delta}{1-\theta\delta}\big]\frac{(1-\theta)f(x)h_0(x)}{\theta
  +(1-\theta)f(x)}[\theta
  +(1-\theta)f(x)]dx\\
&=\frac{1-\theta}{\theta(1-\theta\delta)}\int_0^1f(x)h_0(x)\mathbf{1}_{[0
  , 1 -  \delta)}(x)dx-\frac{(1-\theta)\delta}{1-\theta\delta}\int_0^1f(x)h_0(x)dx\\
&=  0.
\end{align*}
This establishes~\eqref{equa_projection}. Let us now calculate the efficient information 
\begin{align*}
\tilde{I}_{\theta,f} =&\mathbb{P}_{\theta,f}(\tilde{l}_{\theta,f}^2)\\
=&
\int_0^1\Big(\frac{1}{\theta^2}\mathbf{1}_{[1-\delta,
  1]}(x)+\frac{\delta^2}{(1-\theta\delta)^2}\mathbf{1}_{[0, 1-\delta)}(x)\Big)[\theta+(1-\theta)f(x)]dx\\
=&\frac{\delta}{\theta}+\frac{\delta^2}{(1-\theta\delta)^2}(1-\theta\delta)\\
=& \frac{\delta}{\theta(1-\theta\delta)}.
\end{align*}

We now turn to the particular  case where $\delta=0$. In this case the
previous  computations show  that $\dot  l_{\theta,f}$ belongs  to the
closure of the linear span of ${}_f\mathcal{\dot P}_{\delta}$ and that
the Fisher information is zero. 
Now, we show that the functional $\psi(\mathbb{P}_{\theta,f})=\theta$ is not differentiable at
$\mathbb{P}_{\theta,f}$ relative to the tangent set
$\mathcal{\dot{P}}_{0}=\text{lin}\big(\dot{l}_{\theta,
  f}+{}_f\mathcal{\dot{P}}_{0}\big)= {}_f\dot{\mathcal{P}}_{0}$. In fact, if this were true, there would
exist a function $\tilde{\psi}_{\theta,f}$ such that
\[
a= 
\frac{\partial}{\partial t}
\psi(\mathbb{P}_{\theta+ta,f_t}) _{\big|_{t=0} } =\langle \tilde{\psi}_{\theta,
  f},a\dot{l}_{\theta,f}+h\rangle,\ \forall a \in \mathbb{R}, h \in
{}_f\mathcal{\dot{P}}_{0},
\]
where  $\langle   \cdot,\cdot  \rangle$  denotes   scalar  product  in
$\mathbb{L}^2(\mathbb{P}_{\theta,     f})$.      Choosing     $h     =
-\dot{l}_{\theta,f}\in     {}_f\mathcal{\dot{P}}_{0}$,    we    obtain
$a=(a-1)\langle   \tilde{\psi}_{\theta,  f},\dot{l}_{\theta,f}\rangle$
for every value $a \in \mathbb{R}$, 
which is impossible. \\


For the rest of the proof, we set $\delta>0$.
Using Lemma 25.25 in \cite{VanderVaart1998}, we remark that the functional
$\psi(\mathbb{P}_{\theta,f})=\theta$ is differentiable at
$\mathbb{P}_{\theta,f}$ relative to the tangent set
$\mathcal{\dot{P}_{\delta}}$
with efficient influence function given by
\begin{align*}
\tilde{\psi}_{\theta,f}(x) =&
\tilde{I}_{\theta,f}^{-1}\tilde{l}_{\theta,f}(x)\\
=&\frac{\theta(1-\theta\delta)}{\delta}\Big(\frac{1}{\theta}\mathbf{1}_{[1-\delta,1]}(x)-\frac{\delta}{1-\theta\delta}\mathbf{1}_{[0,
  1-\delta)}(x)\Big)\\
=&\frac{1-\theta\delta}{\delta}\mathbf{1}_{[1-\delta,
  1]}(x)-\theta\mathbf{1}_{[0, 1-\delta)}(x)\\
=&\frac{1}{\delta}\mathbf{1}_{[1-\delta, 1]}(x)-\theta,
\end{align*}
which concludes the proof.
\end{proof}


\begin{proof}[Proof of Corollary~\ref{cor:LAM}]
We start by dealing with the case  $\delta = 0$. Let us recall that in
this  case,   the  ordinary  score  $\dot   l_{\theta,f}$  belongs  to
${}_f\mathcal{\dot P}_0$. 
 We first remark that this tangent set 
${}_f\mathcal{\dot P}_0$ is a
linear subspace of $\mathbb{L}^2(\mathbb{P}_{\theta,f})$ with infinite
dimension. So we can choose an orthonormal basis 
$\{h_i\}_{i=1}^{\infty}$ of ${}_f \dot{\mathcal{P}}_{0}$ such
that for every $m$, we have $\dot{l}_{\theta,f} \notin
{}_f\dot{\mathcal{P}}_{0,m} := \text{lin}(h_1,h_2,\ldots, h_m)$. We thus have
 \begin{multline*}
\underset{E_0}{\sup}\ \underset{n \rightarrow \infty}{\liminf}\ \underset{g
  \in E_0}{\sup}\ \esp _{P_{1/\sqrt{n},g}}
\big[\sqrt{n}\big(\theta_n-\psi(P_{1/\sqrt{n},g})\big)\big]^2 \\ \geq
\underset{F_0}{\sup}\ \underset{n \rightarrow \infty}{\liminf}\ \underset{g
  \in F_0}{\sup}\ \esp _{P_{1/\sqrt{n},g}}
\big[\sqrt{n}\big(\theta_n-\psi(P_{1/\sqrt{n},g})\big)\big]^2,
\end{multline*}
where $E_0$ and $F_0$ range through all finite subsets of the tangent sets
$\dot{\mathcal{P}}_{0}=\text{lin}(\dot   l_{\theta,f}+{}_f\mathcal{\dot
  P}_0)={}_f\mathcal{\dot P}_0$ and $\text{lin}\big(\dot{l}_{\theta,f}+{}_f\dot{\mathcal{P}}_{0,m}\big)= {}_f\dot{\mathcal{P}}_{0,m}$,
respectively. The
efficient score function for $\theta$ corresponding to the tangent set
${}_f\dot{\mathcal{P}}_{0,m}$ is
\begin{equation*}
\tilde{l}_{\theta,f,m} = \dot{l}_{\theta,f} - \sum_{i=1}^m \langle
\dot{l}_{\theta,f},h_i \rangle h_i \neq 0. 
\end{equation*}
Moreover, the efficient information  $\tilde{I}_{\theta,f,m}
=\mathbb{P}_{\theta,f}(\tilde{l}_{\theta,f,m}^2)$                    is
non zero. Using Lemma 25.25 from \cite{VanderVaart1998}, we remark that the functional
$\psi(\mathbb{P}_{\theta,f})=\theta$ is differentiable at
$\mathbb{P}_{\theta,f}$ relative to the tangent set
$\text{lin}\big(\dot{l}_{\theta,f}+{}_f\dot{\mathcal{P}}_{0,m}\big)$
with efficient influence function $\tilde{\psi}_{\theta,f,m}=
\tilde{I}_{\theta,f,m}^{-1}\tilde{l}_{\theta,f,m}$.  So  we can  apply
Theorem 25.21 from 
\cite{VanderVaart1998} to obtain that
\begin{equation*}
\underset{F_0}{\sup}\ \underset{n \rightarrow \infty}{\liminf}\ \underset{g
  \in F_0}{\sup}\ \esp _{P_{1/\sqrt{n},g}}
\big[\sqrt{n}\big(\theta_n-\psi(P_{1/\sqrt{n},g})\big)\big]^2 \geq \tilde{I}_{\theta,f,m}^{-1}.
\end{equation*}
Since $\tilde{I}_{\theta,f,m}
\xrightarrow[m\to\infty]  {} \tilde  I_{\theta,f}= 0$,  we  obtain the
result. 
The second part of the proof concerning $\delta>0$ is an immediate consequence of Theorem~\ref{thm:I_and_psi} together with Theorem 25.21 and Lemma 25.23 in \cite{VanderVaart1998}.
\end{proof}

\subsection{Proofs from Sections~\ref{sec:his-estimator} and \ref{sec:one_step}}\label{sec:proof_upper}

\begin{proof}[Proof of Proposition~\ref{prop:sqrtn_consistent}]
Let us denote by  $\mathcal{D} = \{ 1, 2, \cdots, D\}$, $\mathcal{D}_0 = \{ k \in \mathcal{D}\ \text{such that}\ I_k
\subseteq [1-\delta, 1]\}$ and $\mathcal{D}_1= \mathcal{D} \setminus \mathcal{D}_0 =\{ k \in \mathcal{D}\ \text{such that}\ I_k
\nsubseteq [1-\delta, 1]\}$. We start by proving that the estimator
$\hat \theta_{I,n}$ converges almost surely to $\theta$. Indeed, we
can write that
\begin{equation}\label{eq:hat_theta_decomp}
\hat \theta_{I,n} =\theta +  \sum_{k \in \mathcal{D}_0}\Big(\frac{n_k}{n|I_k|} -\theta\Big) \mathbf{1}\{\hat  k_n=k\} + 
  (\hat\theta_{I,n} -\theta)\mathbf{1}\{I_{\hat k_n} \nsubseteq [1-\delta, 1]\} ,
\end{equation}
where $\mathbf{1}\{A\}$ or $\mathbf{1}_A$ is used to denote the indicator function of set $A$.
By using the strong law of large numbers, we have the almost sure convergences 
\begin{align*}
\forall k \in \mathcal{D}_0, \quad & \frac
{n_k}{n|I_k|} \xrightarrow[n\to  +\infty]{a.s.} \theta ,\\
\forall k \in \mathcal{D}_1, \quad & \frac
{n_k}{n|I_k|}  \xrightarrow[n\to  +\infty]{a.s.}  \frac{\alpha_k}{|I_k|}=\frac 1  {|I_k|}\int_{I_k}
g(u)du >\theta .
\end{align*}
As a consequence, we obtain that the second term in the right-hand side of~\eqref{eq:hat_theta_decomp} converges almost surely to zero, namely
\[
\Big|\sum_{k \in \mathcal{D}_0}\Big(\frac{n_k}{n|I_k|} -\theta\Big)
\mathbf{1}\{\hat  k_n=k\}\Big| \leq \sum_{k \in
  \mathcal{D}_0}\Big|\frac{n_k}{n|I_k|} -\theta\Big| \xrightarrow[n\to +\infty]{a.s.} 0 .
\]
The third term in the right-hand side of~\eqref{eq:hat_theta_decomp}
also converges almost surely to zero. Indeed, we have
\[
|\hat
\theta_{I,n}  -\theta|\mathbf{1}\{I_{\hat k_n}  \nsubseteq  [1-\delta,
  1]\} \le \Big(\max_{1\le   k\le   D}   \frac{1}{|I_k|}-\theta\Big) \sum_{k \in \mathcal{D}_1}\mathbf{1}\{\hat  k_n=k\}.
\]
We fix an integer $k_0 \in \mathcal{D}_0$, then for all $k \in
\mathcal{D}_1$, we have
\begin{eqnarray*}
\mathbf{1}\{\hat  k_n=k\} &=&\mathbf{1}\Big\{\frac{n_k}{n|I_k|} \leq
\frac{n_j}{n|I_j|},\ \forall j \in \mathcal{D}\Big\}\\
&\leq& \mathbf{1}\Big\{\frac{n_k}{n|I_k|} \leq
\frac{n_{k_0}}{n|I_{k_0}|} \Big\}\\
&\leq& \mathbf{1}\Big\{\frac{n_{k_0}}{n|I_{k_0}|} -\theta
+\frac{\alpha_k}{|I_k|}-\frac{n_k}{n|I_k|} \geq \frac{\alpha_k}{|I_k|} - \theta\Big\}.
\end{eqnarray*}
Since $\epsilon_k =\alpha_k/|I_k| - \theta >0$ and 
\[
\frac{n_{k_0}}{n|I_{k_0}|} -\theta
+\frac{\alpha_k}{|I_k|}-\frac{n_k}{n|I_k|}  \xrightarrow[n\to  +\infty]{a.s.} 0,
\]
we obtain that
\[
\mathbf{1}\Big\{\frac{n_{k_0}}{n|I_{k_0}|} -\theta
+\frac{\alpha_k}{|I_k|}-\frac{n_k}{n|I_k|} \geq \epsilon_k \Big\} \xrightarrow[n\to  +\infty]{a.s.} 0,
\]
which concludes the proof of the almost sure convergence of $\hat \theta_{I,n}$.

We now write 
\begin{equation}\label{eq:hat_theta_decomp_sqrtn}
\sqrt{n} (\hat \theta_{I,n} -\theta) = \sum_{k \in \mathcal{D}_0}\sqrt{n} \Big(\frac{n_k}{n|I_k|} -\theta\Big)
\mathbf{1}\{\hat  k_n=k\} +\sqrt{n}  (\hat
\theta_{I,n} -\theta)\mathbf{1}\{I_{\hat k_n} \nsubseteq [1-\delta, 1]\} .
\end{equation}
The  second term  in  the  right hand-side  of  the previous  equation
converges in  probability to zero.  Indeed, for any $\epsilon  >0$, we
have
\begin{eqnarray*}
\pr(\sqrt{n}  |\hat
\theta_{I,n}  -\theta|\mathbf{1}\{I_{\hat k_n}  \nsubseteq  [1-\delta, 1]\}
\ge \epsilon) &\le& \pr (I_{\hat k_n} \nsubseteq [1-\delta, 1]) \\
&\le& \sum_{k \in \mathcal{D}_1} \pr \Big(\frac{n_{k_0}}{n|I_{k_0}|} -\theta
+\frac{\alpha_k}{|I_k|}-\frac{n_k}{n|I_k|} \geq \epsilon_k\Big) \xrightarrow[n\to +\infty]{} 0.
\end{eqnarray*}
Now, whenever  $k \in \mathcal{D}_0$, by denoting $$\sigma_k^2 =
\theta\Big(\frac{1}{|I_k|} - \theta\Big),$$ the  central limit theorem
gives the convergence in distribution
\[
\sqrt{n}    \Big(\frac{n_k}{n|I_k|}   -\theta\Big)   \xrightarrow[n\to
+\infty]{d}  \mathcal{N} \Big(  0, \sigma_k^2
\Big) .
\]
As   a   consequence,   each    of   these   terms   is   bounded   in
probability. According to~\eqref{eq:hat_theta_decomp_sqrtn}, we conclude 
\[
\sqrt{n} (\hat \theta_{I,n} -\theta) =O_{\mathbb{P}}(1).
\]
We now prove the third statement of the proposition. We have
\[
 \Var(\sqrt{n}\hat\theta_{I,n}) \le
 \esp\big[(\sqrt{n}(\hat\theta_{I,n}-\theta))^2\big],
\]
where
\begin{equation}\label{eq:decomp_variance}
\esp\big[(\sqrt{n}(\hat\theta_{I,n}-\theta))^2\big] = \sum_{k \in
  \mathcal{D}_0} \esp\Big[\Big(\sqrt{n}\big(\frac{n_k}{n|I_k|}
-\theta\big)\Big)^2 \mathbf{1}_{\hat  k_n=k} \Big] + \sum_{k \in
  \mathcal{D}_1} \esp\Big[\Big(\sqrt{n}\big(\frac{n_k}{n|I_k|} -\theta\big)\Big)^2 \mathbf{1}_{\hat  k_n=k} \Big].
\end{equation}
The second term in the right-hand side of~\eqref{eq:decomp_variance}
is bounded by
\begin{equation*}
\sum_{k \in
  \mathcal{D}_1} \esp\Big[\Big(\sqrt{n}\big(\frac{n_k}{n|I_k|}
-\theta\big)\Big)^2 \mathbf{1}_{\hat  k_n=k} \Big] \leq
\Big(\max_{1\leq k \leq D} \frac{1}{|I_k|} -\theta\Big)^2 \sum_{k \in
  \mathcal{D}_1} n \pr (\hat  k_n=k),
\end{equation*}
where for all $k \in \mathcal{D}_1$, fixing an integer $k_0 \in \mathcal{D}_0$ and 
according to Hoeffding's inequality, 
\begin{eqnarray*}
\pr (\hat  k_n=k) 
&\leq&\pr\Big(\frac{n_k}{n|I_k|} \leq
\frac{n_{k_0}}{n|I_{k_0}|} \Big)\\
&\leq&\pr\Big[\sum_{i=1}^n\Big(\frac{1}{|I_{k_0}|}\mathbf{1}\{ X_i \in
I_{k_0}\} -\theta
+\frac{\alpha_k}{|I_k|}-\frac{1}{|I_k|}\mathbf{1}\{ X_i \in
I_k\}  \Big)\geq n \epsilon_k\Big]\\
&\leq& \exp \Big[-2n\epsilon_k^2\Big(\frac{1}{|I_k|}+\frac{1}{|I_{k_0}|}\Big)^{-2} \Big].
\end{eqnarray*}
Thus, we get that
\begin{multline*}
\sum_{k \in
  \mathcal{D}_1} \esp\Big[\Big(\sqrt{n}\big(\frac{n_k}{n|I_k|}
-\theta\big)\Big)^2 \mathbf{1}_{\hat  k_n=k} \Big] \\
\leq
\Big(\max_{1\leq k \leq D} \frac{1}{|I_k|} -\theta\Big)^2 \sum_{k \in
  \mathcal{D}_1} n\exp
\Big[-2n\epsilon_k^2\Big(\frac{1}{|I_k|}+\frac{1}{|I_{k_0}|}\Big)^{-2}
\Big] \xrightarrow[n\to +\infty]{} 0.
\end{multline*}
For the first term in the right-hand side
of~\eqref{eq:decomp_variance}, we apply Cauchy-Schwarz's inequality
\begin{eqnarray}\label{eq:var_bounded1}
\sum_{k \in
  \mathcal{D}_0} \esp\Big[\Big(\sqrt{n}\big(\frac{n_k}{n|I_k|}
-\theta\big)\Big)^2 \mathbf{1}_{\hat  k_n=k} \Big] &\leq& \sqrt{\sum_{k \in
  \mathcal{D}_0} \esp\Big[\Big(\sqrt{n}\big(\frac{n_k}{n|I_k|}
-\theta\big)\Big)^4\Big] } \sqrt{\sum_{k \in \mathcal{D}_0} \pr(\hat
k_n=k)} \nonumber\\
&\leq& \sqrt{\sum_{k \in
  \mathcal{D}_0} \esp\Big[\Big(\sqrt{n}\big(\frac{n_k}{n|I_k|}
-\theta\big)\Big)^4\Big] },
\end{eqnarray}
where for all $k \in \mathcal{D}_0$, 
\begin{eqnarray}\label{eq:var_bounded2}
& &\esp\Big[\Big(\sqrt{n}\big(\frac{n_k}{n|I_k|}
-\theta\big)\Big)^4\Big] = \esp\Big[\frac{1}{n^2}\Big( \sum_{i=1}^n \big( \frac{1}{|I_k|}\mathbf{1}\{ X_i \in
I_k\} -\theta \big)\Big)^4\Big] \nonumber\\
&=& \frac{1}{n}  \esp\Big[\Big(\frac{1}{|I_k|}\mathbf{1}\{ X_1 \in
I_k\} -\theta \Big)^4\Big] + \frac{n-1}{n} \esp^2\Big[\Big(\frac{1}{|I_k|}\mathbf{1}\{ X_1 \in
I_k\} -\theta \Big)^2\Big] \nonumber \\
&=&
\frac{\theta}{n}\Big(\frac{1}{|I_k|^3}-\frac{4\theta}{|I_k|^2}+\frac{6\theta^2}{|I_k|}-3\theta^3\Big)
+ \frac{n-1}{n} \sigma_k^4.
\end{eqnarray}
Thus, we finally obtain that
\begin{multline*}
 \Var(\sqrt{n}\hat\theta_{I,n}) \le \sqrt{\sum_{k \in
  \mathcal{D}_0} \Big[\frac{\theta}{n}\Big(\frac{1}{|I_k|^3}-\frac{4\theta}{|I_k|^2}+\frac{6\theta^2}{|I_k|}-3\theta^3\Big)
+ \frac{n-1}{n} \sigma_k^4\Big] }\ + \\ \Big(\max_{1\leq k \leq D} \frac{1}{|I_k|} -\theta\Big)^2 \sum_{k \in
  \mathcal{D}_1} n\exp
\Big[-2n\epsilon_k^2\Big(\frac{1}{|I_k|}+\frac{1}{|I_{k_0}|}\Big)^{-2}
\Big] \xrightarrow[n\to +\infty]{} \sqrt{\sum_{k \in
  \mathcal{D}_0} \sigma_k^4 }.
\end{multline*}
\end{proof}

\begin{proof}[Proof of Proposition~\ref{prop:equiv_efficiency}]
Let us first establish that condition~\eqref{equa-int} holds. 
In fact, with the notation $p_{\theta,f} = \theta+(1-\theta)f$,
we have
\begin{eqnarray*}
&&\int \Arrowvert \tilde{l}_{\theta_n,f}d\mathbb{P}_{\theta_n,f}^{1/2}
-\tilde{l}_{\theta,f}d\mathbb{P}_{\theta,f}^{1/2}\Arrowvert ^2
=\int_0^1\Big(\tilde{l}_{\theta_n,f}(x)\sqrt{p_{\theta_n,f}(x)} -
\tilde{l}_{\theta,f}(x)\sqrt{p_{\theta,f}(x)}\Big)^2 dx\\
&\leq&2\int_0^1\big(\tilde{l}_{\theta_n,f}(x) -
\tilde{l}_{\theta,f}(x)\big)^2 p_{\theta_n,f}(x) dx+2\int_0^1
\tilde{l}_{\theta,f}^2(x)\Big(\sqrt{p_{\theta_n,f}(x)}
-\sqrt{p_{\theta,f}(x)}\Big)^2 dx\\
&\leq&2\int_0^1\Big[\frac{1}{\theta_n}-\frac{1}{\theta}+\Big(\frac{1}{\theta(1-\theta\delta)}-\frac{1}{\theta_n(1-\theta_n\delta)}\Big)\mathbf{1}_{\{f(x)>0\}}\Big]^2p_{\theta_n,f}(x)dx\\
&& + 2\int_0^1\Big[\frac{1}{\theta}-\frac{1}{\theta(1-\theta\delta)}\mathbf{1}_{\{f(x)>0\}}\Big]^2\frac{(\theta_n-\theta)^2(1-f(x))^2}{\Big(\sqrt{p_{\theta_n,f}(x)}
+\sqrt{p_{\theta,f}(x)}\Big)^2}dx\\
&\leq&2\int_0^1(\theta_n-\theta)^2\Big[\frac{1}{\theta\theta_n}+\frac{\delta(\theta+\theta_n)+1}{\theta\theta_n(1-\theta\delta)(1-\theta_n\delta)}\mathbf{1}_{\{f(x)>0\}}\Big]^2p_{\theta_n,f}(x)dx\\
&& + 2\int_0^1(\theta_n-\theta)^22\Big[\frac{1}{\theta^2}+\frac{1}{\theta^2(1-\theta)^2}\Big]\frac{(1-f(x))^2}{\big(\sqrt{\theta_n}
+\sqrt{\theta}\big)^2}dx\\
&\leq&(\theta_n-\theta)^2\Big[\frac{C}{\theta^2}+\frac{C(1+2C\theta}{\theta^2(1-\theta)^2}\Big]^2
+ C(\theta_n-\theta)^2\Big[\frac{1}{\theta^3}+\frac{1}{\theta^3(1-\theta)^2}\Big]=O\big(\frac{1}{n}\big),
\end{eqnarray*}
where $C$ is some positive constant. Thus, according to  Theorem
7.4 from \cite{VanderVaartlecture}, the existence of an asymptotically efficient sequence of
estimators of $\theta$ is equivalent to
the existence of a sequence of estimators $\hat{l}_{n,\theta}$
satisfying~\eqref{equa-vitesse-para}
and~\eqref{equa-variance}. 

Now in model $\mathcal{P}_\delta$, the efficient score
function $\tilde{l}_{\theta,f}$ is given by
\begin{equation*}
\tilde{l}_{\theta,f}(x)                                               =
\frac{1}{\theta}-\frac{1}{\theta(1-\theta\delta)}\mathbf{1}_{[0, 1 - \delta)}(x) ,
\end{equation*}
so that it  is natural to estimate the parameter  $\delta$ in order to
estimate  $\tilde{l}_{\theta,f}$. Let  $\hat{\delta}_n$  be any  given
consistent (in probability) estimator of
$\delta$. Let us examine condition~\eqref{equa-vitesse-para} more closely. We have 
\begin{align*}
\sqrt{n}\mathbb{P}_{\theta_n,f}\hat{l}_{n,\theta_n} =& \sqrt{n}\mathbb{P}_{\theta_n,f}(\hat{l}_{n,\theta_n} -\tilde{l}_{\theta_n,f})\\
=& \sqrt{n}\int_0^1 \frac{1}{\theta_n}
\Big[\frac{1}{1-\theta_n\hat{\delta}_n} \mathbf{1}_{[0, 1 -
    \hat{\delta}_n)}(x) -\frac{1}{1-\theta_n\delta}\mathbf{1}_{[0, 1
  - \delta)}(x)\Big] g_{\theta_n, f}(x) dx\\
=&\int_0^1 \frac{\sqrt{n}}{\theta_n}
\Big[\Big(\frac{1}{1-\theta_n\hat{\delta}_n }- \frac{1}{1-\theta_n\delta}\Big) \mathbf{1}_{[0, 1 -
    \hat{\delta}_n)}(x) \\
 & + \frac{1}{1-\theta_n\delta}\Big(\mathbf{1}_{[0, 1 -
    \hat{\delta}_n)}(x) - \mathbf{1}_{[0, 1
  - \delta)}(x)\Big)\Big] g_{\theta_n, f}(x) dx\\
=& \sqrt{n}(\hat{\delta}_n- \delta) \int_0^{1-\hat{\delta}_n} 
\frac{g_{\theta_n, f}(x)}{(1-\theta_n\delta)(1-\theta_n\hat{\delta}_n)} dx +
  \sqrt{n}\int_{1-\delta}^{1-\hat{\delta}_n} \frac{g_{\theta_n, f}(x)
  }{1-\theta_n\delta} dx\\
=& \sqrt{n}(\hat{\delta}_n- \delta) \Big[\int_0^{1-\delta} 
\frac{g_{\theta, f}(x)}{(1-\theta\delta)^2} dx -\frac{g_{\theta, f}(1-\delta)
  }{1-\theta\delta}  + o_{\mathbb{P}}(1)\Big].
\end{align*}

Hence, the "no-bias" condition~\eqref{equa-vitesse-para} is equivalent to the
existence of an estimator $\hat{\delta}_n$ of $\delta$ that
 converges at a rate faster than $1/\sqrt{n}$, namely such that  $ \sqrt{n}(\hat{\delta}_n- \delta) =
o_{\mathbb{P}}(1)$. With the same argument as in the previous 
calculation, the  consistency condition~\eqref{equa-variance} is 
satisfied as soon as the estimator $\hat{\delta}_n$ converges in
probability to $\delta$. 
\end{proof}

\subsection{Proof of Proposition~\ref{prop:CR_sqrtn_consistent}}\label{sec:proof_CR}

For each partition $I$, let us denote by $\mathcal{F}_I $ the vector
space of piecewise constant functions built from the partition $I$ and
$g_I$ the orthogonal projection of $g \in L^2([0,1])$ onto
$\mathcal{F}_I $. The mean squared error of a histogram
estimator $\hat{g}_I$ can be written as the sum of a bias term and
a variance term
\begin{equation*}
\esp[||g-\hat{g}_I||_2^2]=||g-g_I||_2^2+\esp[||g_I-\hat{g}_I||_2^2].
\end{equation*}
We introduce three lemmas that are needed to prove Proposition~\ref{prop:CR_sqrtn_consistent}. The proofs of these technical lemmas is further postponed to Appendix~\ref{appendix}.

\begin{lem}
\label{lem1}
Let $I = (I_k)_{k=1}^D$ be an arbitrary partition of $[0,1]$.
Then the variance term of the mean squared error of a histogram
estimator $\hat{g}_I$ is bounded by $C/n$, where $C$ is a positive
constant. In other words,  
\begin{equation*}
\esp[||g_I-\hat{g}_I||_2^2] = O\big(\frac{1}{n}\big).
\end{equation*}
\end{lem} 
For any partition  $I=(I_k)_{1,\ldots,D}$ of $[0,1]$, we let  
\[
L(I) = ||g_I-g||_2^2 \quad \mbox{and}\quad \hat{L}_p(I) = \hat{R}_p(I) + ||g||_2^2,
\]
respectively the bias term of the mean squared error of a histogram
estimator $\hat{g}_I$ and its estimator. 
\begin{lem}
\label{lem2}
Let $I=(I_k)_{1,\ldots,D}$ be an arbitrary partition of $[0,1]$. 
Let $p \in \{1,2, \ldots, n-1\}$ such that $\underset{n\rightarrow
  \infty}{\lim}p/n< 1$. Then we have the following results\\
i) $\hat{L}_p(I) \xrightarrow[n\to\infty] {a.s.} L(I)$\\
ii) $\sqrt{n}\big(\hat{L}_p(I) -L(I)\big)=\sqrt{n}\big(\hat{R}_p(I)
-R(I)\big) + \frac{1}{\sqrt{n}}(s_{11} -s_{21})
\xrightarrow[n\to\infty] {d} \mathcal{N}(0, 4\sigma_I^2)$, where 
\[
\sigma_I^2 = s_{32} - s_{21}^2\ \text{with}\ s_{ij} = \displaystyle\sum_k\frac{\alpha_k^i}{|I_k|^j} , \forall (i,j)
\in \mathbb{N}^2.
\]
\end{lem}
Let $I, J$ be two partitions in $\mathcal{I}$, then $I$ is called a
subdivision of $J$ and we denote $I \trianglelefteq J$, if
$\mathcal{F}_J \subset \mathcal{F}_I$ and $I \ntrianglelefteq J$ otherwise.

\begin{lem}
\label{lem3}
 Suppose that function $f$ satisfies Assumption~\ref{assum1}. Let us consider
$m_{max}$ large enough such that $\mu^\star - \lambda^\star >
2^{1-m_{max}}$. Define $N = 2^{m_{max}}$ and $I^{(N)} =
(N,\lambda_N,\mu_N) \in \mathcal{I}$ with $\lambda_N = \lceil N\lambda^\star
\rceil /N$, $\mu_N = \lfloor N\mu^\star \rfloor /N$. Then for every
partition $I \in \mathcal{I}$, we have\\
i) If $I$ is a subdivision of $I^{(N)}$, then $L(I) =
L(I^{(N)})$.\\
ii) If $I$ is not a subdivision of $I^{(N)}$, then $L(I) >
L(I^{(N)})$.
\end{lem}

We are now ready to prove Proposition~\ref{prop:CR_sqrtn_consistent}, starting by establishing point $i)$.
First, we remark that under condition~\eqref{condition-sij},
\citeauthor{Celisse-Robin2010} prove in their Proposition~2.1 that 
$$
\frac{\hat{p}(I)}{n}\xrightarrow[n\to\infty] {a.s.} l_{\infty}(I) \in [0,1).
$$
Denoting by $\Lambda^\star=[\lambda^\star,\mu^\star]$ and $\hat{\Lambda}=[\hat{\lambda},\hat{\mu}]$, we may write
\begin{eqnarray}
\hat{\theta}_n^{CR} &=&
\theta+(\hat{\theta}_n^{CR}-\theta)\mathbf{1}_{\hat{I}  \trianglelefteq I^{(N)}}+(\hat{\theta}_n^{CR}-\theta)\mathbf{1}_{\hat{I}  \ntrianglelefteq I^{(N)}}
\notag \\
 &=&
\theta+\displaystyle\sum_{I=(N,\lambda,\mu) \trianglelefteq I^{(N)}}\left[\frac{1}{n(\mu-\lambda)}\displaystyle\sum_{i=1}^n\mathbf{1}\{X_i\in
[\lambda,\mu]\}-\theta\right]\mathbf{1}\{\hat{\lambda}=\lambda,\hat{\mu}=\mu\} \notag
\\
&&\
+(\hat{\theta}_n^{CR}-\theta)\mathbf{1}_{\hat{I}  \ntrianglelefteq I^{(N)}}, \label{eq:estimator}
\end{eqnarray}
where $N = 2^{m_{max}}$ as in Lemma ~\ref{lem3}. For each partition $I=(N,\lambda,\mu) \trianglelefteq I^{(N)}$, we have
$[\lambda,\mu]\subseteq\Lambda^\star$. By applying the strong law of large
numbers we get that
\begin{equation*}
\frac{1}{n(\mu-\lambda)}\displaystyle\sum_{i=1}^n\mathbf{1}\{X_i\in
[\lambda,\mu]\}\xrightarrow[n\to\infty] {a.s.} \frac{\mathbb{P}(X_i \in
[\lambda, \mu])}{\mu -\lambda} = \theta.
\end{equation*}
Since the cardinality $card(\mathcal I)$ of $\mathcal{I}$ is finite and does not depend on
$n$, in order to finish the proof, it is sufficient to establish that
$$(\hat{\theta}_n^{CR}-\theta)\mathbf{1}_{\hat{I}  \ntrianglelefteq I^{(N)}}\xrightarrow[n\to\infty]
{a.s.} 0.$$
Using Lemma~\ref{lem3}, we have
$L(\hat{I}) > L(I^{(N)})$. Let 
\begin{equation}
\label{defi-gama}
\gamma = \underset{I\ntrianglelefteq
  I^{(N)}}{\min}L(I) -L(I^{(N)}) >0,
\end{equation}
we obtain that
\begin{multline*}
|\hat{\theta}_n^{CR}-\theta|\mathbf{1}_{\hat{I}  \ntrianglelefteq I^{(N)}} \leq
(N-\theta)\mathbf{1}\{L(\hat{I}) - L(I^{(N)}) \geq \gamma\}
\leq \\(N-\theta)\mathbf{1}\{|\hat{L}_{\hat{p}(\hat{I})}(\hat{I}) - L(\hat{I})|+|\hat{L}_{\hat{p}(I^{N})}(I^{N}) - L(I^{N})|+\hat{L}_{\hat{p}(\hat{I})}(\hat{I})- \hat{L}_{\hat{p}(I^{(N)})}(I^{(N)}) \geq \gamma\}\\
\leq (N-\theta)\mathbf{1}\{2\underset{I\in \mathcal{I}}{\sup}|\hat{L}_{\hat{p}(I)}(I) - L(I)|+\hat{L}_{\hat{p}(\hat{I})}(\hat{I})- \hat{L}_{\hat{p}(I^{(N)})}(I^{(N)}) \geq \gamma\}.
\end{multline*}
By definition of $\hat{I}$, we have
$\hat{L}_{\hat{p}(\hat{I})}(\hat{I})-
\hat{L}_{\hat{p}(I^{(N)})}(I^{(N)}) \leq 0$, so that
\begin{eqnarray}
\label{eq-sup}
|\hat{\theta}_n^{CR}-\theta|\mathbf{1}_{\hat{I}  \ntrianglelefteq I^{(N)}} &\leq&
(N-\theta)\mathbf{1}\{\underset{I\in
    \mathcal{I}}{\sup}|\hat{L}_{\hat{p}(I)}(I) - L(I)| \geq
  \frac{\gamma}{2}\}\\
&\leq&
(N-\theta) \displaystyle\sum_{I \in \mathcal{I}}\mathbf{1}\{|\hat{L}_{\hat{p}(I)}(I) - L(I)| \geq
  \frac{\gamma}{2}\}.\notag
\end{eqnarray}
Since $\forall I \in \mathcal{I}$, we both have $\hat{L}_p(I)
\xrightarrow[n\to\infty] {a.s.} L(I)$ and $\hat{p}(I)/n
\xrightarrow[n\to\infty] {a.s.} l_{\infty}(I) \in  [0,1)$ as well as the
fact that 
$\hat{R}_p(I)$ (given by~\eqref{LPO-estimator}) is a continuous
function of $p/n$, we obtain $\hat{L}_{\hat{p}(I)}(I)
\xrightarrow[n\to\infty] {a.s.} L(I)$. Therefore,
\begin{equation*}
\mathbf{1}\{|\hat{L}_{\hat{p}(I)}(I) - L(I)| \geq
  \frac{\gamma}{2}\}\xrightarrow[n\to\infty] {a.s.} 0.
\end{equation*}
Indeed, if $X_n \xrightarrow[] {a.s.} X$ then $\forall
\epsilon >0,\ \text{we have}\ \mathbf{1}\{|X_n-X|\geq \epsilon\}\xrightarrow[]
{a.s.} 0$. It thus follows that $(\hat{\theta}_n^{CR}-\theta)\mathbf{1}_{\hat{I}  \ntrianglelefteq I^{(N)}}\xrightarrow[]
{a.s.} 0$. We finally get that $\hat{\theta}_n^{CR} \xrightarrow[] {a.s.} \theta$.

We now turn to point $ii)$. 
We may write as previously, 
\begin{eqnarray*}
\sqrt{n}(\hat{\theta}_n^{CR} - \theta) =
\displaystyle\sum_{I=(N,\lambda,\mu) \trianglelefteq I^{(N)}}&&\sqrt{n}\Big[\frac{1}{n(\mu-\lambda)}\displaystyle\sum_{i=1}^n\mathbf{1}\{X_i\in
[\lambda,\mu]\}-\theta\Big]\mathbf{1}_{\{\hat{\lambda}=\lambda,\hat{\mu}=\mu\}}\\
& &+\sqrt{n}(\hat{\theta}_n^{CR}-\theta)\mathbf{1}_{\{\hat{I}  \ntrianglelefteq I^{(N)}\}}.
\end{eqnarray*}
For each partition $I=(N,\lambda,\mu) \trianglelefteq I^{(N)}$, by
applying the central limit theorem, we get that
\begin{equation*}
\sqrt{n}\big[\frac{1}{n(\mu-\lambda)}\displaystyle\sum_{i=1}^n\mathbf{1}_{X_i\in
[\lambda,\mu]}-\theta\big] \xrightarrow[n\to\infty] {d} \mathcal{N}\Big(0,\theta\Big(\frac{1}{\mu-\lambda}-\theta\Big)\Big).
\end{equation*}
Hence, using again that $\text{card} (\mathcal{I})$ is finite,
\begin{equation}
\label{eq-somme}
\displaystyle\sum_{I=(N,\lambda,\mu) \trianglelefteq I^{(N)}}\sqrt{n}\big[\frac{1}{n(\mu-\lambda)}\displaystyle\sum_{i=1}^n\mathbf{1}_{X_i\in
[\lambda,\mu]}-\theta\big]\mathbf{1}_{\hat{\lambda}=\lambda,\hat{\mu}=\mu}
= O_{\mathbb{P}}(1).
\end{equation}
We shall now prove that
$\sqrt{n}(\hat{\theta}_n^{CR}-\theta)\mathbf{1}_{\hat{I}  \ntrianglelefteq I^{(N)}}\xrightarrow[n\to\infty]
{\mathbb{P}} 0$. In fact, according to~\eqref{eq-sup}, for all $\epsilon > 0$, we have
\begin{eqnarray*}
\mathbb{P}(\sqrt{n}|\hat{\theta}_n^{CR}-\theta|\mathbf{1}_{\hat{I}  \ntrianglelefteq I^{(N)}}
>\epsilon) &\leq& \mathbb{P}(\hat{I}  \ntrianglelefteq I^{(N)})\\
&\leq& \mathbb{P}(\underset{I\in
    \mathcal{I}}{\sup}|\hat{L}_{\hat{p}(I)}(I) - L(I)| \geq
  \frac{\gamma}{2})\\
 &\leq&\displaystyle\sum_{I \in \mathcal{I}}\mathbb{P}(|\hat{L}_{\hat{p}(I)}(I) - L(I)| \geq
  \frac{\gamma}{2}) \xrightarrow[n\to\infty] {} 0,
\end{eqnarray*}
where $\gamma$ is defined by~\eqref{defi-gama}. Therefore,
$\sqrt{n}(\hat{\theta}_n^{CR}-\theta)\mathbf{1}_{\hat{I}  \ntrianglelefteq I^{(N)}}
=o_{\mathbb{P}}(1)$. We finally conclude that $\sqrt{n}(\hat{\theta}_n^{CR} -
\theta) = O_{\mathbb{P}}(1).$\\
We now prove the last statement $iii)$ of the proposition. We have
\[
 \Var(\sqrt{n}\hat\theta_{n}^{CR}) \le
 \esp\big[(\sqrt{n}(\hat\theta_{n}^{CR}-\theta))^2\big],
\]
where
\begin{eqnarray*}\label{eq:decomp_variance2}
\esp\big[(\sqrt{n}(\hat\theta_{n}^{CR}-\theta))^2\big] =
\displaystyle\sum_{I=(N,\lambda,\mu) \trianglelefteq I^{(N)}}&&\esp
\Big[\frac{1}{n}\Big(\displaystyle\sum_{i=1}^n \big(\frac{1}{\mu-\lambda}\mathbf{1}\{X_i\in
[\lambda,\mu]\}-\theta\big)\Big)^2\mathbf{1}_{\{\hat{\lambda}=\lambda,\hat{\mu}=\mu\}}\Big]\\
& &+\esp\big[(\sqrt{n}(\hat\theta_{n}^{CR}-\theta))^2 \mathbf{1}\{ \hat{I}\ntrianglelefteq I^{(N)}\}\big]. 
\end{eqnarray*}
The first term of the above equation is bounded as in the proof of Proposition~\ref{prop:sqrtn_consistent} (see inequalities~\eqref{eq:var_bounded1} and \eqref{eq:var_bounded2})
\begin{eqnarray*}
& &\displaystyle\sum_{I=(N,\lambda,\mu) \trianglelefteq I^{(N)}}\esp
\Big[\frac{1}{n}\Big(\displaystyle\sum_{i=1}^n \big(\frac{1}{\mu-\lambda}\mathbf{1}\{X_i\in
[\lambda,\mu]\}-\theta\big)\Big)^2\mathbf{1}_{\{\hat{\lambda}=\lambda,\hat{\mu}=\mu\}}\Big]\\
&\leq& \sqrt{\sum_{I=(N,\lambda,\mu) \trianglelefteq I^{(N)}} \Big[\frac{\theta}{n}\Big(\frac{1}{(\mu-\lambda)^3}-\frac{4\theta}{(\mu-\lambda)^2}+\frac{6\theta^2}{\mu-\lambda}-3\theta^3\Big)
+ \frac{n-1}{n} \theta^2\Big(\frac{1}{(\mu-\lambda)} - \theta\Big)^2\Big] }. 
\end{eqnarray*}
The second term is bounded by
\begin{eqnarray*}
\esp\big[(\sqrt{n}(\hat\theta_{n}^{CR}-\theta))^2 \mathbf{1}\{
\hat{I}\ntrianglelefteq I^{(N)}\}\big] &\leq& (N-\theta)^2n\mathbb{P}(\hat{I}\ntrianglelefteq I^{(N)})\\
&\leq& (N-\theta)^2n\mathbb{P}(\underset{I\in
    \mathcal{I}}{\sup}|\hat{L}_p(I) - L(I)| \geq
  \frac{\gamma}{2})\\
 &\leq&(N-\theta)^2n\displaystyle\sum_{I \in \mathcal{I}}\mathbb{P}(|\hat{L}_p(I) - L(I)| \geq
  \frac{\gamma}{2}).
\end{eqnarray*}
For each partition $I \in \mathcal{I}$, according to the calculations
in the proof of Lemma \ref{lem1}, we have
\begin{eqnarray*}
\hat{L}_p(I)-L(I) 
&=& \frac{2n-p}{(n-1)(n-p)}\displaystyle\sum_k\frac{n_k}{n|I_k|}
-
\frac{n(n-p+1)}{(n-1)(n-p)}\displaystyle\sum_k\frac{1}{|I_k|}\big(\frac{n_k}{n}\big)^2
+ s_{21}\\
& =&\frac{2n-p}{(n-1)(n-p)} \Big\{\displaystyle\sum_k\frac{1}{|I_k|}\Big(\frac{n_k}{n} -\alpha_k\Big) + s_{11}-s_{21}\Big\}\\ 
& -&\frac{n(n-p+1)}{(n-1)(n-p)}\displaystyle\sum_k\frac{1}{|I_k|}\Big(\frac{n_k}{n}
-\alpha_k\Big)^2 \\
&-&\frac{2n(n-p+1)}{(n-1)(n-p)}\displaystyle\sum_k\frac{\alpha_k}{|I_k|}\Big(\frac{n_k}{n}
-\alpha_k\Big).
\end{eqnarray*}
This leads to
\begin{eqnarray*}
\mathbb{P}(|\hat{L}_p(I) - L(I)| \geq  \frac{\gamma}{2})
&\leq& \mathbb{P}\Big(\Big|\displaystyle\sum_k\frac{1}{|I_k|}\Big(\frac{n_k}{n} -\alpha_k\Big) \Big| \geq \frac{(n-1)(n-p)\gamma}{6(2n-p)}
-|s_{21}-s_{11}|\Big)\\ 
&+& \mathbb{P}\Big(\displaystyle\sum_k\frac{1}{|I_k|}\Big(\frac{n_k}{n}
-\alpha_k\Big)^2 \geq \frac{(n-1)(n-p)\gamma}{6n(n-p+1)} \Big)\\
&+&\mathbb{P}\Big(\Big|\displaystyle\sum_k\frac{\alpha_k}{|I_k|}\Big(\frac{n_k}{n}
-\alpha_k\Big) \Big| \geq \frac{(n-1)(n-p)\gamma}{12n(n-p+1)}
\Big).
\end{eqnarray*}
According to Hoeffding's inequality, we have
\begin{eqnarray*}
&&\mathbb{P}\Big(\Big|\displaystyle\sum_k\frac{1}{|I_k|}\Big(\frac{n_k}{n}
-\alpha_k\Big) \Big| \geq \frac{(n-1)(n-p)\gamma}{6(2n-p)}
-|s_{21}-s_{11}|\Big)\\
&=&\mathbb{P}\Big(\Big|\displaystyle\sum_{i=1}^n\displaystyle\sum_k\frac{1}{|I_k|}\Big(\mathbf{1}\{X_i
\in I_k\}-\alpha_k\Big) \Big| \geq \frac{n(n-1)(n-p)\gamma}{6(2n-p)}
-n|s_{21}-s_{11}|\Big)\\
&\leq& 2\exp\Big[-2n\Big(\displaystyle\sum_k\frac{1}{|I_k|}\Big)^{-2}\Big(\frac{(n-1)(n-p)\gamma}{6(2n-p)}
-|s_{21}-s_{11}|\Big)^2\Big],
\end{eqnarray*} 
as well as 
\begin{equation*}
\mathbb{P}\Big(\Big|\displaystyle\sum_k\frac{\alpha_k}{|I_k|}\Big(\frac{n_k}{n}
-\alpha_k\Big) \Big| \geq \frac{(n-1)(n-p)\gamma}{12n(n-p+1)}\Big)
\leq 2\exp\Big[-2n s_{11}^{-2}\Big(\frac{(n-1)(n-p)\gamma}{12n(n-p+1)}
\Big)^2\Big],
\end{equation*} 
and
\begin{eqnarray*}
&&\mathbb{P}\Big(\Big|\displaystyle\sum_k\frac{1}{|I_k|}\Big(\frac{n_k}{n}
-\alpha_k\Big)^2 \Big| \geq \frac{(n-1)(n-p)\gamma}{6n(n-p+1)}
\Big)\\
&\leq&\displaystyle\sum_k\mathbb{P}\Big(\Big|\frac{n_k}{n}
-\alpha_k \Big|^2 \geq \frac{|I_k| (n-1)(n-p)\gamma}{6Dn(n-p+1)}
\Big)\\
&\leq&\displaystyle\sum_k\mathbb{P}\Big(\Big|\displaystyle\sum_{i=1}^n\Big(\mathbf{1}\{X_i
\in I_k\}-\alpha_k\Big) \Big|^2 \geq \frac{|I_k| n(n-1)(n-p)\gamma}{6D(n-p+1)}
\Big)\\
&\leq& 2\exp\Big[-2\Big(\frac{|I_k| (n-1)(n-p)\gamma}{6D(n-p+1)}
\Big)\Big].
\end{eqnarray*}
Hence, we obtain that $n \mathbb{P}(|\hat{L}_p(I) - L(I)| \geq
\frac{\gamma}{2}) \xrightarrow[n\to +\infty]{} 0$. Finally, we
conclude that $\underset{n \rightarrow \infty}{\limsup}
  \Var(\sqrt{n}\hat\theta^{CR}_{n}) < +\infty$. 

\
\\
\
\begin{merci}
 The authors are grateful to Cyril Dalmasso, Elisabeth Gassiat and Pierre Neuvial for
 fruitful discussions concerning this work. They also would like to thank an anonymous referee whose remarks and suggestions greatly improved this work.
\end{merci}

\bibliographystyle{chicago}
\bibliography{prop_true_null}

\def\cprime{$'$}
\begin{thebibliography}{}

\bibitem[\protect\citeauthoryear{Benjamini and Hochberg}{Benjamini and
  Hochberg}{1995}]{Benjamini1995}
Benjamini, Y. and Y.~Hochberg (1995).
\newblock Controlling the false discovery rate: a practical and powerful
  approach to multiple testing.
\newblock {\em J. Roy. Statist. Soc. Ser. B\/}~{\em 57\/}(1), 289--300.

\bibitem[\protect\citeauthoryear{Broberg}{Broberg}{2005}]{Broberg2005}
Broberg, P. (2005).
\newblock A comparative review of estimates of the proportion unchanged genes
  and the false discovery rate.
\newblock {\em BMC Bioinformatics\/}~{\em 6\/}(1), 199.

\bibitem[\protect\citeauthoryear{Cai and Jin}{Cai and Jin}{2010}]{Cai_Jin10}
Cai, T.~T. and J.~Jin (2010).
\newblock Optimal rates of convergence for estimating the null density and
  proportion of nonnull effects in large-scale multiple testing.
\newblock {\em Ann. Statist.\/}~{\em 38\/}(1), 100--145.

\bibitem[\protect\citeauthoryear{Celisse and Robin}{Celisse and
  Robin}{2008}]{Celisse-Robin2008}
Celisse, A. and S.~Robin (2008).
\newblock Nonparametric density estimation by exact leave-{$p$}-out
  cross-validation.
\newblock {\em Comput. Statist. Data Anal.\/}~{\em 52\/}(5), 2350--2368.

\bibitem[\protect\citeauthoryear{Celisse and Robin}{Celisse and
  Robin}{2010}]{Celisse-Robin2010}
Celisse, A. and S.~Robin (2010).
\newblock A cross-validation based estimation of the proportion of true null
  hypotheses.
\newblock {\em J. Statist. Plann. Inference\/}~{\em 140\/}(11), 3132--3147.

\bibitem[\protect\citeauthoryear{Chamberlain}{Chamberlain}{1986}]{Chamberlain1986}
Chamberlain, G. (1986).
\newblock Asymptotic efficiency in semiparametric models with censoring.
\newblock {\em J. Econometrics\/}~{\em 32\/}(2), 189--218.

\bibitem[\protect\citeauthoryear{Dudoit and van~der Laan}{Dudoit and van~der
  Laan}{2008}]{Dudoit_book}
Dudoit, S. and M.~J. van~der Laan (2008).
\newblock {\em Multiple testing procedures with applications to genomics}.
\newblock Springer Series in Statistics. New York: Springer.

\bibitem[\protect\citeauthoryear{Efron}{Efron}{2004}]{Efron2004}
Efron, B. (2004).
\newblock Large-scale simultaneous hypothesis testing: the choice of a null
  hypothesis.
\newblock {\em J. Amer. Statist. Assoc.\/}~{\em 99\/}(465), 96--104.

\bibitem[\protect\citeauthoryear{Efron, Tibshirani, Storey, and Tusher}{Efron
  et~al.}{2001}]{Efron2001}
Efron, B., R.~Tibshirani, J.~D. Storey, and V.~Tusher (2001).
\newblock Empirical {B}ayes analysis of a microarray experiment.
\newblock {\em J. Amer. Statist. Assoc.\/}~{\em 96\/}(456), 1151--1160.

\bibitem[\protect\citeauthoryear{Genovese and Wasserman}{Genovese and
  Wasserman}{2004}]{Genovese2004}
Genovese, C. and L.~Wasserman (2004).
\newblock A stochastic process approach to false discovery control.
\newblock {\em Ann. Statist.\/}~{\em 32\/}(3), 1035--1061.

\bibitem[\protect\citeauthoryear{Hengartner and Stark}{Hengartner and
  Stark}{1995}]{Hengartner1995}
Hengartner, N.~W. and P.~B. Stark (1995).
\newblock Finite-sample confidence envelopes for shape-restricted densities.
\newblock {\em Ann. Statist.\/}~{\em 23\/}(2), 525--550.

\bibitem[\protect\citeauthoryear{Ibragimov and
  Has{\cprime}minski{\u\i}}{Ibragimov and
  Has{\cprime}minski{\u\i}}{1981}]{Ibra}
Ibragimov, I.~A. and R.~Z. Has{\cprime}minski{\u\i} (1981).
\newblock {\em Statistical estimation}, Volume~16 of {\em Applications of
  Mathematics}.
\newblock New York: Springer-Verlag.
\newblock Asymptotic theory, Translated from the Russian by Samuel Kotz.

\bibitem[\protect\citeauthoryear{Jin}{Jin}{2008}]{Jin08}
Jin, J. (2008).
\newblock Proportion of non-zero normal means: universal oracle equivalences
  and uniformly consistent estimators.
\newblock {\em J. R. Stat. Soc. Ser. B Stat. Methodol.\/}~{\em 70\/}(3),
  461--493.

\bibitem[\protect\citeauthoryear{Jin and Cai}{Jin and Cai}{2007}]{Jin_Cai07}
Jin, J. and T.~Cai (2007).
\newblock Estimating the null and the proportion of nonnull effects in
  large-scale multiple comparisons.
\newblock {\em J. Amer. Statist. Assoc.\/}~{\em 102\/}(478), 495--506.

\bibitem[\protect\citeauthoryear{Langaas, Lindqvist, and Ferkingstad}{Langaas
  et~al.}{2005}]{Langaas2005}
Langaas, M., B.~H. Lindqvist, and E.~Ferkingstad (2005).
\newblock Estimating the proportion of true null hypotheses, with application
  to {DNA} microarray data.
\newblock {\em J. R. Stat. Soc. Ser. B Stat. Methodol.\/}~{\em 67\/}(4),
  555--572.

\bibitem[\protect\citeauthoryear{Meinshausen and B{\"u}hlmann}{Meinshausen and
  B{\"u}hlmann}{2005}]{Meinshausen_Buhlmann05}
Meinshausen, N. and P.~B{\"u}hlmann (2005).
\newblock Lower bounds for the number of false null hypotheses for multiple
  testing of associations under general dependence structures.
\newblock {\em Biometrika\/}~{\em 92\/}(4), 893--907.

\bibitem[\protect\citeauthoryear{Meinshausen and Rice}{Meinshausen and
  Rice}{2006}]{Meinshausen_Rice}
Meinshausen, N. and J.~Rice (2006).
\newblock Estimating the proportion of false null hypotheses among a large
  number of independently tested hypotheses.
\newblock {\em Ann. Statist.\/}~{\em 34\/}(1), 373--393.

\bibitem[\protect\citeauthoryear{Mosig, Lipkin, Khutoreskaya, Tchourzyna,
  Soller, and Friedmann}{Mosig et~al.}{2001}]{Mosigetal}
Mosig, M.~O., E.~Lipkin, G.~Khutoreskaya, E.~Tchourzyna, M.~Soller, and
  A.~Friedmann (2001).
\newblock A whole genome scan for quantitative trait loci affecting milk
  protein percentage in israeli-holstein cattle, by means of selective milk dna
  pooling in a daughter design, using an adjusted false discovery rate
  criterion.
\newblock {\em Genetics\/}~{\em 157\/}(4), 1683--1698.

\bibitem[\protect\citeauthoryear{Nettleton, Hwang, Caldo, and Wise}{Nettleton
  et~al.}{2006}]{Nettleton06}
Nettleton, D., J.~Hwang, R.~Caldo, and R.~Wise (2006).
\newblock Estimating the number of true null hypotheses from a histogram of p
  values.
\newblock {\em Journal of Agricultural, Biological, and Environmental
  Statistics\/}~{\em 11}, 337--356.

\bibitem[\protect\citeauthoryear{Neuvial}{Neuvial}{2010}]{Neuvial2010}
Neuvial, P. (2010).
\newblock Intrinsic bounds and false discovery rate control in multiple testing
  problems.
\newblock Technical report, arXiv:1003.0747.

\bibitem[\protect\citeauthoryear{Schweder and Spj\o{}tvoll}{Schweder and
  Spj\o{}tvoll}{1982}]{Schweder1982}
Schweder, T. and E.~Spj\o{}tvoll (1982).
\newblock Plots of p-values to evaluate many tests simultaneously.
\newblock {\em Biometrika\/}~{\em 69\/}(3), 493--502.

\bibitem[\protect\citeauthoryear{Storey}{Storey}{2002}]{Storey2002}
Storey, J.~D. (2002).
\newblock A direct approach to false discovery rates.
\newblock {\em J. R. Stat. Soc. Ser. B Stat. Methodol.\/}~{\em 64\/}(3),
  479--498.

\bibitem[\protect\citeauthoryear{Storey and Tibshirani}{Storey and
  Tibshirani}{2003}]{Storey2003}
Storey, J.~D. and R.~Tibshirani (2003).
\newblock Statistical significance for genomewide studies.
\newblock {\em Proc. Natl. Acad. Sci. USA\/}~{\em 100\/}(16), 9440--9445
  (electronic).

\bibitem[\protect\citeauthoryear{Turkheimer, Smith, and Schmidt}{Turkheimer
  et~al.}{2001}]{Turkheimer2001}
Turkheimer, F., C.~Smith, and K.~Schmidt (2001).
\newblock Estimation of the number of true null hypotheses in multivariate
  analysis of neuroimaging data.
\newblock {\em NeuroImage\/}~{\em 13\/}(5), 920 -- 930.

\bibitem[\protect\citeauthoryear{van~der Vaart}{van~der
  Vaart}{2002}]{VanderVaartlecture}
van~der Vaart, A. (2002).
\newblock {Semiparametric statistics.}
\newblock {Bolthausen, Erwin et al., Lectures on probability theory and
  statistics. Ecole d'\'et\'e de probabilit\'es de Saint-Flour XXIX - 1999,
  Saint-Flour, France, July 8-24, 1999. Berlin: Springer. Lect. Notes Math.
  1781, 331-457 (2002).}

\bibitem[\protect\citeauthoryear{van~der Vaart}{van~der
  Vaart}{1998}]{VanderVaart1998}
van~der Vaart, A.~W. (1998).
\newblock {\em Asymptotic statistics}, Volume~3 of {\em Cambridge Series in
  Statistical and Probabilistic Mathematics}.
\newblock Cambridge: Cambridge University Press.

\end{thebibliography}


\newpage
\appendix                    

\section{Appendix. Proofs of technical lemmas}            
   \label{appendix}
\subsection{Proof of Lemma~\ref{lem1}}
Note that \cite{Celisse-Robin2010} prove that
$\esp[||g-\hat{g}_I||_2^2] \xrightarrow[n\to\infty] {} 0$, while we
further establish that it is $O(1/n)$. By a simple bias-variance decomposition, we may write
\begin{equation*}
\esp[||g_I-\hat{g}_I||_2^2] = \esp[||g-\hat{g}_I||_2^2] -
||g_I-g||_2^2.
\end{equation*}
As for the bias term, it is easy to show that
\begin{eqnarray}
\label{term-bias}
||g-g_I||_2^2 &=&\underset{h \in \mathcal{F}_I}{\inf}||g-h||_2^2
\notag\\
&=&\underset{(a_k)_k \in \mathbb{R}}{\inf}\Big[||g||_2^2
-2\int_0^1\big(\displaystyle\sum_ka_k\mathbf{1}_{I_k}(x)\big)g(x)dx + \int_0^1\big(\displaystyle\sum_ka_k\mathbf{1}_{I_k}(x)\big)^2dx\Big] \notag\\
&=&\underset{(a_k)_k \in \mathbb{R}}{\inf}\Big[||g||_2^2-2
\displaystyle\sum_ka_k\alpha_k +\displaystyle\sum_ka_k^2|I_k|\Big] \notag\\
&=&||g||_2^2- \displaystyle\sum_k\frac{\alpha_k^2}{|I_k|} = ||g||_2^2
- s_{21}.
\end{eqnarray}
Let us now calculate the mean squared error of $\hat{g}_I$
\begin{eqnarray*}
\esp[||g-\hat{g}_I||_2^2] &=& ||g||_2^2+\esp\Big[||\hat{g}_I||_2^2 -
2\int_0^1\hat{g}_I(x)g(x)dx\Big]\\
&=&||g||_2^2+\esp\Big[\int_0^1\big(\sum_k\frac{n_k}{n|
  I_k|}\mathbf{1}_{I_k}(x)\big)^2dx -
2\int_0^1\sum_k\frac{n_k}{n|
  I_k|}\mathbf{1}_{I_k}(x)g(x)dx\Big]\\
&=&||g||_2^2+\esp\Big[\sum_k\frac{n_k^2}{n^2|
  I_k|} -
2\sum_k\frac{n_k\alpha_k}{n|
  I_k|}\Big] .
\end{eqnarray*}
Since $n_k $ follows a Binomial distribution $ \mathcal{B}(n, \alpha_k)$, we have
\begin{equation*}
\esp[n_k] = n\alpha_k \ \text{and} \ \esp[n_k^2] = n^2\alpha_k^2 + n\alpha_k(1-\alpha_k).
\end{equation*}
Therefore,
\begin{eqnarray}
\label{term-MSE}
\esp[||g-\hat{g}_I||_2^2]&=& ||g||_2^2+ \sum_k\frac{n^2\alpha_k^2+n\alpha_k(1-\alpha_k)}{n^2|
  I_k|} -2 \sum_k\frac{n\alpha_k^2}{n|I_k|}\notag\\
&=& ||g||_2^2 - s_{21} + \frac{1}{n}(s_{11} -s_{21}).
\end{eqnarray}
Using~\eqref{term-bias}  and~\eqref{term-MSE}, we  obtain  the desired
result, namely 
\begin{equation*}
\esp[||g_I-\hat{g}_I||_2^2] = \esp[||g-\hat{g}_I||_2^2] -
||g_I-g||_2^2 = \frac{1}{n}(s_{11} -s_{21}) = O\big(\frac{1}{n}\big).
\end{equation*}

\subsection{Proof of Lemma~\ref{lem2}}
i) Since 
\begin{equation*}
\underset{n \rightarrow \infty}{\lim} \frac{p}{n} < 1\ \text{and}\ \frac{n_k}{n}
\xrightarrow[n\to\infty] {a.s.} \alpha_k,\ \text{for all}\ k,
\end{equation*}
we obtain that
\begin{eqnarray*}
\hat{L}_p(I)&=&||g||_2^2+\frac{2n-p}{(n-1)(n-p)}\displaystyle\sum_k\frac{n_k}{n|I_k|}
-
\frac{n(n-p+1)}{(n-1)(n-p)}\displaystyle\sum_k\frac{1}{|I_k|}\big(\frac{n_k}{n}\big)^2\\
&\xrightarrow[n\to\infty] {a.s.}&||g||_2^2 -
\displaystyle\sum_k\frac{\alpha_k^2}{|I_k|}=||g||_2^2 -s_{21} =||g_I-g||_2^2= L(I).
\end{eqnarray*}
ii) By definition of $R(I)$ and using~\eqref{term-MSE}, we have
\begin{equation*}
R(I) = \esp[||g-\hat{g}_I||_2^2] - ||g||_2^2 = - s_{21} + \frac{1}{n}(s_{11} -s_{21}).
\end{equation*}
This gives that
\begin{eqnarray}
\label{equa-ii}
\sqrt{n}[\hat{R}_p(I)-R(I)] &=& \sqrt{n}\Big[\frac{2n-p}{(n-1)(n-p)}\displaystyle\sum_k\frac{n_k}{n|I_k|}
-
\frac{n(n-p+1)}{(n-1)(n-p)}\displaystyle\sum_k\frac{1}{|I_k|}\big(\frac{n_k}{n}\big)^2
\notag\\ & &+ s_{21} - \frac{1}{n}(s_{11} -s_{21})\Big] \notag\\
&=&\frac{2n-p}{(n-1)(n-p)}\displaystyle\sum_k\frac{1}{|I_k|}\big[\sqrt{n}\big(\frac{n_k}{n}
-\alpha_k\big)\big] +\frac{(2n-p)\sqrt{n}}{(n-1)(n-p)}s_{11} \notag\\ & &-\frac{n(n-p+1)}{\sqrt{n}(n-1)(n-p)}\displaystyle\sum_k\frac{1}{|I_k|}\big[\sqrt{n}\big(\frac{n_k}{n}
-\alpha_k\big)\big]^2 -\frac{(2n-p)\sqrt{n}}{(n-1)(n-p)}s_{21}\notag\\ & &-\frac{2n(n-p+1)}{(n-1)(n-p)}\displaystyle\sum_k\frac{\alpha_k}{|I_k|}\big[\sqrt{n}\big(\frac{n_k}{n}
-\alpha_k\big)\big]-\frac{1}{\sqrt{n}}(s_{11} -s_{21}) \notag\\
&=&T_1 -\frac{2n(n-p+1)}{(n-1)(n-p)}\displaystyle\sum_k\frac{\alpha_k}{|I_k|}\big[\sqrt{n}\big(\frac{n_k}{n}
-\alpha_k\big)\big ].
\end{eqnarray}
Then, using the central limit theorem and the continuity of the function $x
\mapsto x^2$, we have
\begin{equation*}
\sqrt{n}\big(\frac{n_k}{n} -\alpha_k\big) \xrightarrow[n\to\infty] {d} \mathcal{N}(0,\alpha_k(1-\alpha_k)),
\end{equation*}
\begin{equation*}
\big[\sqrt{n}\big(\frac{n_k}{n} -\alpha_k\big)\big]^2
\xrightarrow[n\to\infty] {d} Z_k^2 \ \text{with}\ Z_k \sim \mathcal{N}(0,\alpha_k(1-\alpha_k)).
\end{equation*}
It thus follows that $T_1=o_{\mathbb{P}}(1)$.
We now consider the remaining term in~\eqref{equa-ii}. We have 
\begin{eqnarray*}
\displaystyle\sum_k\frac{\alpha_k}{|I_k|}\big[\sqrt{n}\big(\frac{n_k}{n}
-\alpha_k\big)\big]&=&\frac{1}{\sqrt{n}}\displaystyle\sum_k\frac{\alpha_k}{|I_k|}n_k
-\sqrt{n}\displaystyle\sum_k\frac{\alpha_k^2}{|I_k|}\\
&=&\frac{1}{\sqrt{n}}\displaystyle\sum_k\frac{\alpha_k}{|I_k|}\big(\sum_{i=1}^n\mathbf{1}_{X_i \in
  I_k}\big)
-\sqrt{n}\ s_{21}\\
&=&\frac{1}{\sqrt{n}}\sum_{i=1}^n\big(\displaystyle\sum_k\frac{\alpha_k}{|I_k|}\mathbf{1}_{X_i \in
  I_k}-s_{21}\big).
\end{eqnarray*}
Let us denote
\begin{equation*}
Y_i=\displaystyle\sum_k\frac{\alpha_k}{|I_k|}\mathbf{1}_{X_i \in
  I_k}-s_{21}.
\end{equation*}
Then the random variables $Y_1, Y_2, \ldots, Y_n$ are iid centered with variance
\begin{equation*}
\sigma_I^2=\esp(Y_1^2)=\esp\Big(\displaystyle\sum_k\frac{\alpha_k^2}{|I_k|^2}\mathbf{1}_{X_1 \in
  I_k}-2s_{21}\displaystyle\sum_k\frac{\alpha_k}{|I_k|}\mathbf{1}_{X_1 \in
  I_k} + s_{21}^2\Big)=s_{32}-s_{21}^2.
\end{equation*}
By the central limit theorem, we obtain 
\begin{equation*}
\displaystyle\sum_k\frac{\alpha_k}{|I_k|}\big[\sqrt{n}\big(\frac{n_k}{n}
-\alpha_k\big)\big]\xrightarrow[n\to\infty] {d} \mathcal{N}(0, \sigma_I^2).
\end{equation*}
Combining this with~\eqref{equa-ii} implies that
\begin{equation*}
 \sqrt{n}[\hat{R}_p(I)-R(I)]\xrightarrow[n\to\infty] {d} \mathcal{N}(0, 4\sigma_I^2).
\end{equation*}
It is easy to calculate that
\begin{equation*}
 \sqrt{n}\big(\hat{L}_p(I) -L(I)\big)=\sqrt{n}\big(\hat{R}_p(I)
-R(I)\big) + \frac{1}{\sqrt{n}}(s_{11} -s_{21}).
\end{equation*}
Hence, we have
\begin{equation*}
 \sqrt{n}[\hat{L}_p(I)-L(I)]\xrightarrow[n\to\infty] {d} \mathcal{N}(0, 4\sigma_I^2),
\end{equation*}
which completes the proof.
\subsection{Proof of Lemma~\ref{lem3}}
i) If $I$ is a subdivision of $I^{(N)}$, then $I = (N,\lambda,
\mu)$ with $[\lambda, \mu] \subset [\lambda^\star, \mu^\star]$. For example,
we may have the following situation\\
\setlength{\unitlength}{0.9mm}
\begin{picture}(20,15)
 \put(10,0){\line(1,0){123}} 
 \put(30,-1){\line(0,1){2}} 
 \put(130,-1){\line(0,1){2}}
\put(40,-1){\line(0,1){2}}
\put(50,-1){\line(0,1){2}}
\put(60,-1){\line(0,1){2}}
\put(70,-1){\line(0,1){2}}
\put(110,-1){\line(0,1){2}}
\put(120,-1){\line(0,1){2}}
\put(64,0){\circle*{1}}
\put(117,0){\circle*{1}}
\put(30,5){$0$}
\put(70,5){$\lambda_{N}$}
\put(107,5){$\mu_{N}$}
\put(62,5){$\lambda^\star$}
\put(116,5){$\mu^\star$}
\put(130,5){$1$}
\put(140,0){$I^{(N)}$}
\end{picture}\\
\setlength{\unitlength}{0.9mm}
\begin{picture}(20,15)
 \put(10,0){\line(1,0){123}} 
 \put(30,-1){\line(0,1){2}} 
 \put(130,-1){\line(0,1){2}}
\put(40,-1){\line(0,1){2}}
\put(50,-1){\line(0,1){2}}
\put(60,-1){\line(0,1){2}}
\put(70,-1){\line(0,1){2}}
\put(80,-1){\line(0,1){2}}
\put(110,-1){\line(0,1){2}}
\put(120,-1){\line(0,1){2}}
\put(64,0){\circle*{1}}
\put(117,0){\circle*{1}}
\put(30,5){$0$}
\put(70,5){$\lambda_{N}$}
\put(107,5){$\mu_{N}$}
\put(80,5){$\lambda$}
\put(100,5){$\mu$}
\put(100,-1){\line(0,1){2}}
\put(62,5){$\lambda^\star$}
\put(116,5){$\mu^\star$}
\put(130,5){$1$}
\put(140,0){$I$}
\end{picture}\\
\\
Since $g$ is constant on the interval $[\lambda^\star,\mu^\star]\supset
[\lambda_N,\mu_N]\supset [\lambda,\mu]$, we have $g_I = g_{I^{(N)}} = g$
on the interval $[\lambda_N,\mu_N]$. This implies that $||g_I-g||_2^2 =
||g_{I^{(N)}}-g||_2^2$.\\
ii) If $I =(2^m,\lambda, \mu)$ is not a subdivision of $I^{(N)}$,
then there are two cases to consider:\\
If $m=m_{max}$ then $[\lambda,\mu] \nsubseteq [\lambda_N,\mu_N]$. For
example, we may have\\
\setlength{\unitlength}{0.9mm}
\begin{picture}(20,15)
 \put(10,0){\line(1,0){123}} 
 \put(30,-1){\line(0,1){2}} 
 \put(130,-1){\line(0,1){2}}
\put(40,-1){\line(0,1){2}}
\put(50,-1){\line(0,1){2}}
\put(60,-1){\line(0,1){2}}
\put(70,-1){\line(0,1){2}}
\put(110,-1){\line(0,1){2}}
\put(120,-1){\line(0,1){2}}
\put(64,0){\circle*{1}}
\put(117,0){\circle*{1}}
\put(30,5){$0$}
\put(70,5){$\lambda_{N}$}
\put(107,5){$\mu_{N}$}
\put(62,5){$\lambda^\star$}
\put(116,5){$\mu^\star$}
\put(130,5){$1$}
\put(140,0){$I^{(N)}$}
\end{picture}\\
\setlength{\unitlength}{0.9mm}
\begin{picture}(20,15)
 \put(10,0){\line(1,0){123}} 
 \put(30,-1){\line(0,1){2}} 
 \put(130,-1){\line(0,1){2}}
\put(40,-1){\line(0,1){2}}
\put(50,-1){\line(0,1){2}}
\put(60,-1){\line(0,1){2}}
\put(110,-1){\line(0,1){2}}
\put(120,-1){\line(0,1){2}}
\put(64,0){\circle*{1}}
\put(70,0){\circle*{1}}
\put(117,0){\circle*{1}}
\put(30,5){$0$}
\put(70,5){$\lambda_{N}$}
\put(107,5){$\mu_{N}$}
\put(59,5){$\lambda$}
\put(99,5){$\mu$}
\put(100,-1){\line(0,1){2}}
\put(63,5){$\lambda^\star$}
\put(116,5){$\mu^\star$}
\put(130,5){$1$}
\put(140,0){$I$}
\end{picture}\\
\\
Since $g_I = g_{I^{(N)}} = g$ on the interval $[\lambda_N,\mu_N]$ and
the two partitions $I$ and $I^{(N)}$ restricted to the interval $[\lambda,\mu]^c\cap
[\lambda_N,\mu_N]^c$ are the same, we thus have 
\begin{equation*}
||g_I-g||_{2,[\lambda,\mu]^c}^2 = ||g_{I^{(N)}}-g||_{2,[\lambda,\mu]^c}^2,
\end{equation*}
so that
\begin{equation*}
||g_I-g||_2^2 - ||g_{I^{(N)}}-g||_2^2 =
||g_I-g||_{2,[\lambda,\mu]}^2 - ||g_{I^{(N)}}-g||_{2,[\lambda,\mu]}^2.
\end{equation*}
Using the monotonicity of $f$ on the intervals $[0,\lambda^\star]$ and $[\mu^\star,1]$,
we get that 
\begin{equation*}
||g_I-g||_{2,[\lambda,\mu]}^2 >
||g_{I^{(N)}}-g||_{2,[\lambda,\mu]}^2,\ \text{which implies that}\ L(I) >
L(I^{(N)}).
\end{equation*}
If $m<m_{max}$, we may have for example\\
\setlength{\unitlength}{0.9mm}
\begin{picture}(20,15)
 \put(10,0){\line(1,0){123}} 
 \put(30,-1){\line(0,1){2}} 
 \put(130,-1){\line(0,1){2}}
\put(40,-1){\line(0,1){2}}
\put(50,-1){\line(0,1){2}}
\put(60,-1){\line(0,1){2}}
\put(70,-1){\line(0,1){2}}
\put(110,-1){\line(0,1){2}}
\put(120,-1){\line(0,1){2}}
\put(64,0){\circle*{1}}
\put(117,0){\circle*{1}}
\put(30,5){$0$}
\put(70,5){$\lambda_{N}$}
\put(107,5){$\mu_{N}$}
\put(62,5){$\lambda^\star$}
\put(116,5){$\mu^\star$}
\put(130,5){$1$}
\put(140,0){$I^{(N)}$}
\end{picture}\\
\setlength{\unitlength}{0.9mm}
\begin{picture}(20,15)
 \put(10,0){\line(1,0){123}} 
 \put(30,-1){\line(0,1){2}} 
 \put(130,-1){\line(0,1){2}}
\put(50,-1){\line(0,1){2}}
\put(90,-1){\line(0,1){2}}
\put(110,-1){\line(0,1){2}}
\put(64,0){\circle*{1}}
\put(70,0){\circle*{1}}
\put(117,0){\circle*{1}}
\put(30,5){$0$}
\put(70,5){$\lambda_{N}$}
\put(107,5){$\mu_{N}$}
\put(49,5){$\lambda$}
\put(89,5){$\mu$}
\put(63,5){$\lambda^\star$}
\put(116,5){$\mu^\star$}
\put(130,5){$1$}
\put(140,0){$I$}
\end{picture}\\
\\
As before, we may show that
\begin{equation*}
||g_I-g||_2^2 - ||g_{I^{(N)}}-g||_2^2 \geq
||g_I-g||_{2,[\lambda,\mu]^c}^2 - ||g_{I^{(N)}}-g||_{2,[\lambda,\mu]^c}^2 >0,
\end{equation*}
which completes the proof.

We remark that the assumptions in Lemma 2.1 or Theorem 2.1 in
\cite{Celisse-Robin2010} are not sufficient to show these
results. In fact, the assumption "$g$ is non-constant outside
$\Lambda^\star$" is not sufficient to imply that $\|g -
g_{I^{(N)}}\|_2^2<\|g - g_{\hat{I}}\|_2^2$ in the case where $\hat{I}$
is not  a subdivision of $I^{(N)}$.  For example, let  us consider the
following situation\\
\setlength{\unitlength}{0.9mm}
\begin{picture}(60,15)
 \put(10,0){\line(1,0){123}} 
 \put(30,-1){\line(0,1){2}} 
 \put(130,-1){\line(0,1){2}}
\put(40,-1){\line(0,1){2}}
\put(50,-1){\line(0,1){2}}
\put(60,-1){\line(0,1){2}}
\put(70,-1){\line(0,1){2}}
\put(80,-1){\line(0,1){2}}
\put(110,-1){\line(0,1){2}}
\put(120,-1){\line(0,1){2}}
\put(74,0){\circle*{1}}
\put(117,0){\circle*{1}}
\put(30,5){$0$}
\put(80,5){$\lambda_{N}$}
\put(108,5){$\mu_{N}$}
\put(117,5){$\mu^\star$}
\put(73,5){$\lambda^\star$}
\put(130,5){$1$}
\put(140,0){$I^{(N)}$}
\put(40,5){$a$}
\put(60,5){$b$}
\put(50,5){$c$}
\end{picture}\\
\setlength{\unitlength}{0.9mm}
\begin{picture}(60,15)
 \put(10,0){\line(1,0){123}} 
 \put(30,-1){\line(0,1){2}} 
 \put(130,-1){\line(0,1){2}}
\put(40,-1){\line(0,1){2}}
\put(60,-1){\line(0,1){2}}
\put(70,-1){\line(0,1){2}}
\put(80,-1){\line(0,1){2}}
\put(110,-1){\line(0,1){2}}
\put(120,-1){\line(0,1){2}}
\put(74,0){\circle*{1}}
\put(117,0){\circle*{1}}
\put(30,5){$0$}
\put(80,5){$\lambda_{N}$}
\put(108,5){$\mu_{N}$}
\put(117,5){$\mu^\star$}
\put(74,5){$\lambda^\star$}
\put(130,5){$1$}
\put(140,0){$\hat{I}$}
\put(90,-1){\line(0,1){2}}
\put(100,-1){\line(0,1){2}}
\put(50,6){$\hat{\Lambda}$}
\put(50,5){\vector(0,-1){4}}
\put(40,5){$a$}
\put(60,5){$b$}
\end{picture}\\
\\
We may then calculate that
\begin{equation*}
\|g-g_{\hat{I}}\|_2^2 - \|g-g_{I^{(N)}}\|_2^2 = (c-a)(\alpha_1 - \alpha)^2 + (b-c)(\alpha_2 - \alpha)^2,
\end{equation*}
where
\begin{equation*}
\alpha = \frac{1}{b-a}\int_a^bg(x)dx,\quad \alpha_1 = \frac{1}{c-a}\int_a^cg(x)dx,\quad \alpha_2 = \frac{1}{b-c}\int_c^bg(x)dx.
\end{equation*}
So that if the function $g$ satisfies $\alpha = \alpha_1 = \alpha_2$
(and $g$ is non-constant outside
$\Lambda^\star$) then $\|g - g_{I^{(N)}}\|_2^2 = \|g - g_{\hat{I}}\|_2^2$.

\end{document}